\documentclass[10pt]{article}
\usepackage{amsfonts,color}
 \usepackage{amsmath,amssymb}
 \usepackage{epsfig}
\usepackage{lscape}
\usepackage{amssymb}
\usepackage{graphicx}
 \usepackage{footnote}
\usepackage[utf8]{inputenc}	
\usepackage{mathtools}
\usepackage{amsthm,amssymb,wasysym}
\usepackage{amsfonts}
\usepackage[english]{babel}
\usepackage{tikz}
\usepackage{empheq}
\usepackage{bbm}
\usepackage{stmaryrd}
\usepackage{setspace}
\newcommand{\jump}[1]{ \llbracket #1 \rrbracket }
\newcommand{\id}{{\boldsymbol{\mathbbm{1}}}}

\tikzset{square arrow/.style={to path={-- ++(0,-.25) -| (\tikztotarget)}}}

\newcommand{\Chi}{\raisebox{0.5ex}{\mbox{{\Large $\chi$}}}}

\allowdisplaybreaks[1]

\makeindex

 \newtheorem{theorem}{Theorem}[section]
 
 \newtheorem{remark}[theorem]{Remark}

 \newtheorem{postulate}[theorem]{Postulate}
 \newcommand{\R}{\mathbb{R}}

\DeclareMathOperator{\sym}{sym}
\DeclareMathOperator{\skw}{skew}
\DeclareMathOperator{\tr}{tr}

\DeclareMathOperator{\axl}{axl}
\DeclareMathOperator{\anti}{anti}
\DeclareMathOperator{\dev}{dev}
\DeclareMathOperator{\sL}{\mathfrak{sl}}
\DeclareMathOperator{\so}{\mathfrak{so}}
\DeclareMathOperator{\gl}{\mathfrak{gl}}

\DeclareMathOperator{\Curl}{Curl\,}
\newcommand{\Sym}{ {\rm{Sym}} }

\def\Div{\textrm{Div\,}}
\def\curl{\textrm{curl\,}}

\def\skew{\text{skew}}

\def\dd{\displaystyle}

\makeatletter
\let\@fnsymbol\@arabic

\begin{document}
\title{\large On some fundamental misunderstandings in the indeterminate couple stress model. A comment on the recent papers [A.R. Hadjesfandiari and G.F. Dargush, Couple stress theory for solids, Int. J. Solids  Struct. 48, 2496--2510, 2011; A.R. Hadjesfandiari and G.F. Dargush, Fundamental solutions for isotropic size-dependent couple stress elasticity, Int. J. Solids  Struct. 50, 1253--1265, 2013]}
\author{\normalsize{
Patrizio Neff\thanks{Patrizio Neff,  \ \ Head of Lehrstuhl f\"{u}r Nichtlineare Analysis und Modellierung, Fakult\"{a}t f\"{u}r
Mathematik, Universit\"{a}t Duisburg-Essen,  Thea-Leymann Str. 9, 45127 Essen, Germany, email: patrizio.neff@uni-due.de}\quad
and \quad Ingo M\"unch\thanks{Ingo M\"unch, \ \ Institute for Structural Analysis, Karlsruhe Institute of Technology, Kaiserstr. 12, 76131 Karlsruhe,
Germany, email: ingo.muench@kit.edu} \quad and \quad Ionel-Dumitrel Ghiba\thanks{ Ionel-Dumitrel Ghiba, \ \ \ \ Lehrstuhl f\"{u}r Nichtlineare Analysis und Modellierung, Fakult\"{a}t f\"{u}r Mathematik,
Universit\"{a}t Duisburg-Essen, Thea-Leymann Str. 9, 45127 Essen, Germany;  Alexandru Ioan Cuza University of Ia\c si, Department of Mathematics,  Blvd.
Carol I, no. 11, 700506 Ia\c si,
Romania; Octav Mayer Institute of Mathematics of the
Romanian Academy, Ia\c si Branch,  700505 Ia\c si; and Institute of Solid Mechanics, Romanian Academy, 010141 Bucharest, Romania, {email: dumitrel.ghiba@uni-due.de, dumitrel.ghiba@uaic.ro, tel: 0201-183-2827}} \quad
and \quad
Angela Madeo\footnote{Angela Madeo, \ \  Laboratoire de G\'{e}nie Civil et Ing\'{e}nierie Environnementale,
Universit\'{e} de Lyon-INSA, B\^{a}timent Coulomb, 69621 Villeurbanne
Cedex, France; and International Center M\&MOCS ``Mathematics and Mechanics
of Complex Systems", Palazzo Caetani,
Cisterna di Latina, Italy,
 email: angela.madeo@insa-lyon.fr}}
}
\maketitle

\begin{abstract}

In a series of papers which are either published \cite{hadjesfandiari2013fundamental,hadjesfandiari2011couple} or available as preprints \cite{Dargush,hadjesfandiari2013skew,hadjesfandiari2010polar,hadjesfandiari2014evo} Hadjesfandiari and Dargush
 have reconsidered the linear  indeterminate couple stress model. They are postulating a certain physically plausible split in the virtual work principle. Based on this postulate they claim that the second-order couple stress tensor must always be skew-symmetric. Since they use an incomplete set of boundary conditions in their virtual work principle their statement contains unrecoverable errors.  This is shown by specifying their development to the isotropic case. However, their choice of constitutive parameters is mathematically possible and still yields a well-posed boundary value problem.
\\
\vspace*{0.25cm}
\\
\vspace*{0.25cm}
\end{abstract}

\newpage

\tableofcontents

\section{Introduction}

Among higher gradient elasticity models \cite{Mindlin65,Mindlin68,aifantis2011gradient,maugin1980method,lazar2006note,MauginVirtualPowers} one of the very first models considered in the literature is the so called indeterminate couple stress model \cite{Grioli60,Mindlin62,Toupin64,Koiter64}
in which the higher gradient contributions only enter through gradients
of the continuum rotation, i.e. the total elastic energy can be written
as $$W(\nabla u,\nabla(\nabla u))=W_{\rm lin}({\rm sym}\nabla u)+W_{{\rm curv}}(\nabla{\rm curl}\,u).$$
 In general, higher gradient elasticity models are used to
describe mechanical structures at the micro- and nano-scale or to
regularize certain ill-posed problems by means of these higher gradient
contributions.

In a series of papers which are either published \cite{hadjesfandiari2013fundamental,hadjesfandiari2011couple} or available as preprints \cite{Dargush,hadjesfandiari2013skew,hadjesfandiari2010polar,hadjesfandiari2014evo} Hadjesfandiari and Dargush
 have reconsidered the linear  indeterminate couple stress model. They are postulating a certain physically plausible split in the virtual work principle. Based on this postulate they claim that the second-order couple stress tensor must be skew-symmetric. Since their development has spread considerable confusion  in the field of higher gradient elasticity, we were prompted to carefully re-examine their claim. In doing so we hope to contribute an important clarification in the field and to put an end to the above mentioned confusion.

 In the course of our re-examination  it turned out that the boundary conditions in the classical indeterminate couple stress theory have never been correctly derived.  In \cite{MadeoGhibaNeffMunchCRM,MadeoGhibaNeffMunchKM,NeffGhibaMadeoMunch} we provide, for the first time the consistent and complete boundary conditions for the classical indeterminate couple stress model. In doing so, we find the underlying error in the argument by  Hadjesfandiari and Dargush \cite{hadjesfandiari2011couple,hadjesfandiari2010polar,hadjesfandiari2011couple,hadjesfandiari2013fundamental,hadjesfandiari2013skew,hadjesfandiari2014evo}. While  Hadjesfandiari and Dargush  start with the linear general anisotropic couple stress response and only later specify to isotropy, for definiteness, we consider from the outset the linear isotropic indeterminate couple stress case. Nevertheless, our development is essentially independent of any isotropy assumption.   It is clear that exhibiting the errors in their development for the simpler case of isotropy is sufficient for invalidating their claim.

Before  discussing the papers by Hadjesfandiari and Dargush \cite{hadjesfandiari2011couple,hadjesfandiari2013fundamental,hadjesfandiari2014evo,hadjesfandiari2013skew} we will first recall the indeterminate couple stress model in its accepted format as far as kinematics and equilibrium equations are concerned. We also need to introduce the new set of traction boundary conditions which rectify the shortcomings in all previous papers. For comparison, the up to now accepted traction boundary conditions are also presented.

In the light of this new framework, we try to follow the argument given by  Hadjesfandiari and Dargush \cite{hadjesfandiari2011couple,hadjesfandiari2013fundamental,hadjesfandiari2014evo,hadjesfandiari2013skew} as closely as possible. We will show that their implicitly formulated requirement, recast as a physically plausible postulate by us, leads to the skew-symmetry of the couple stress tensor if and only if  the classical traction boundary conditions are assumed. However, within the corrected format of traction  boundary  conditions no similar conclusion is possible.

Despite the finally erroneous claim by Hadjesfandiari and Dargush we recognize their work in being the reason to reconsider the indeterminate couple stress model and to finally find the underlying error which occurred in the accepted boundary conditions and not directly in    Hadjesfandiari and Dargush's work.

\section{Notational agreements}
In this paper, we denote by $\R^{3\times 3}$ the set of real $3\times 3$ second order tensors, written with
capital letters. For $a,b\in\R^3$ we let $\langle {a},{b}\rangle_{\R^3}$  denote the scalar product on $\R^3$ with
associated vector norm $\|{a}\|^2_{\R^3}=\langle {a},{a}\rangle_{\R^3}$.
The standard Euclidean scalar product on $\R^{3\times 3}$ is given by
$\langle{X},{Y}\rangle_{\R^{3\times3}}=\tr({X Y^T})$, and thus the Frobenius tensor norm is
$\|{X}\|^2=\langle{X},{X}\rangle_{\R^{3\times3}}$. In the following we omit the index
$\R^3,\R^{3\times3}$. The identity tensor on $\R^{3\times3}$ will be denoted by $\id$, so that
$\tr({X})=\langle{X},{\id}\rangle$. We adopt the usual abbreviations of Lie-algebra theory, i.e.,
 $\so(3):=\{X\in\mathbb{R}^{3\times3}\;|X^T=-X\}$ is the Lie-algebra of  skew symmetric tensors
and $\sL(3):=\{X\in\mathbb{R}^{3\times3}\;| \tr({X})=0\}$ is the Lie-algebra of traceless tensors.
 For all $X\in\mathbb{R}^{3\times3}$ we set $\sym X=\frac{1}{2}(X^T+X)\in\Sym$, $\skw X=\frac{1}{2}(X-X^T)\in \so(3)$ and the deviatoric part $\dev X=X-\frac{1}{3}\;\tr(X)\id\in \sL(3)$  and we have
the \emph{orthogonal Cartan-decomposition  of the Lie-algebra} $\gl(3)$
\begin{align}
\gl(3)&=\{\sL(3)\cap \Sym(3)\}\oplus\so(3) \oplus\mathbb{R}\!\cdot\! \id,\quad
X=\dev \sym X+ \skw X+\frac{1}{3}\tr(X) \id\,.
\end{align}
Throughout this paper (when we do not specify else) Latin subscripts take the values $1,2,3$.  Typical conventions for differential
operations are implied such as comma followed
by a subscript to denote the partial derivative with respect to
 the corresponding cartesian coordinate. We also use the Einstein notation of the sum over repeated indices if not differently specified. Here,
we consider the operators $\axl:\so(3)\rightarrow\mathbb{R}^3$ and $\anti:\mathbb{R}^3\rightarrow \so(3)$ through
\begin{align}
(\axl \overline{A})_k=-\frac{1}{2}\, \epsilon_{ijk}\overline{A}_{ij},\qquad \overline{A}.\, v=(\axl \overline{A})\times v, \quad \quad (\anti(v))_{ij}=-\varepsilon_{ijk}v_k, \quad\quad \overline{A}_{ij}=\anti(\axl \overline{A})_{ij},
\quad \quad
\end{align}
 for all  $ v\in\mathbb{R}^3$ and $\overline{A}\in \so(3)$, where $\epsilon_{ijk}$ is the totally antisymmetric third order permutation tensor. We recall that  for a  third  order tensor $\mathbb{E}$ and $X\in \mathbb{R}^{3\times 3}$, $v\in \mathbb{R}^3$ we have the contraction operations $\mathbb{E}: X\in \mathbb{R}^{3}$, $\mathbb{E}. \, v\in \mathbb{R}^{3\times 3}$ and $X.\, v\in \mathbb{R}^3$, with the components
\begin{align}
(\mathbb{E}:\, X)_{i}=\mathbb{E}_{ijk}\,X_{kj}\, , \qquad  (\mathbb{E}. \,v)_{ij}=\mathbb{E}_{ijk}\,v_{k}\,,\qquad (X.\, v)_{i}=X_{ij}\,v_j.
\end{align}
For multiplication of two matrices  we will not use  other specific notations, this means that for $A,B\in \mathbb{R}^{3\times 3}$ we are setting $(A B)_{ij}=A_{ik}B_{kj}$.

 We consider a body which  occupies  a bounded open set $\Omega$ of the three-dimensional Euclidian space $\R^3$ and assume that its boundary
$\partial \Omega$ is a piecewise smooth surface. An elastic material fills the domain $\Omega\subseteq \R^3$ and we refer the motion of the body to  rectangular axes $Ox_i$. Let us consider an open subset $\Gamma$ of $\partial \Omega$.
Here,  $\nu$ is a  vector tangential to the surface $\Gamma$ and which is orthogonal to its boundary $\partial \Gamma$, $\tau=n\times \nu$ is the tangent to the curve $\partial \Gamma$ with respect to the orientation on $\Gamma$.
 We assume that $\partial \Omega$ is a smooth surface. Hence, there are no singularities of the boundary and the jump $\jump{a\cdot \nu}:=[a\cdot \nu]^{+}+[a\cdot \nu]^{-}=([a]^{+}-[a]^{-})\cdot \nu$ of $a$ across the joining
curve $\partial\Gamma$ arises only as consequence of possible discontinuities of the corresponding quantities which follows from the prescribed  boundary conditions on $\Gamma$ and $\partial \Omega\setminus \overline{\Gamma}$, where
$$[\,\cdot\,]^-:=\hspace*{0cm}\dd\lim\limits_{\footnotesize{\begin{array}{c}x\in\partial \Omega\setminus \overline{\Gamma}\\
 \ x\rightarrow \partial \Gamma\end{array}}}\hspace*{0cm} [\,\cdot\,], \qquad \qquad [\,\cdot\,]^+:=\hspace*{-0.2cm}\dd\lim\limits_{\footnotesize{\begin{array}{c}x\in \Gamma\\
 \ x\rightarrow \partial \Gamma\end{array}}}\hspace*{-0.2cm} [\,\cdot\,].$$

  The usual Lebesgue spaces of square integrable functions, vector or tensor fields on $\Omega$ with values in $\mathbb{R}$, $\mathbb{R}^3$ or $\mathbb{R}^{3\times 3}$, respectively will be denoted by $L^2(\Omega)$. Moreover, we introduce the standard Sobolev spaces \cite{Adams75,Raviart79,Leis86}
\begin{align}
\begin{array}{ll}
{\rm H}^1(\Omega)=\{u\in L^2(\Omega)\, |\, {\rm grad}\, u\in L^2(\Omega)\}, &\|u\|^2_{{\rm H}^1(\Omega)}:=\|u\|^2_{L^2(\Omega)}+\|{\rm grad}\, u\|^2_{L^2(\Omega)}\,,\vspace{1.5mm}\\
{\rm H}({\rm curl};\Omega)=\{v\in L^2(\Omega)\, |\, {\rm curl}\, v\in L^2(\Omega)\},   &\|v\|^2_{{\rm H}({\rm curl};\Omega)}:=\|v\|^2_{L^2(\Omega)}+\|{\rm curl}\, v\|^2_{L^2(\Omega)}\, ,\vspace{1.5mm}\\
{\rm H}({\rm div};\Omega)=\{v\in L^2(\Omega)\, |\, {\rm div}\, v\in L^2(\Omega)\}, &\|v\|^2_{{\rm H}({\rm div};\Omega)}:=\|v\|^2_{L^2(\Omega)}+\|{\rm div}\, v\|^2_{L^2(\Omega)}\, ,
\end{array}
\end{align}
of functions $u$ or vector fields $v$, respectively. Furthermore, we introduce their closed subspaces $H_0^1(\Omega)$,  ${\rm H}_0({\rm curl};\Omega)$ as completion under the respective graph norms of the scalar valued space $C_0^\infty(\Omega)$ (the set of infinitely
differentiable functions with compact support in $\Omega$).
  For vector fields $v$ with components in ${\rm H}^{1}(\Omega)$, i.e.
$
v=\left(    v_1, v_2, v_3\right)^T\, , v_i\in {\rm H}^{1}(\Omega),
$
we define
$
 \nabla \,v=\left(
   (\nabla\,  v_1)^T,
    (\nabla\, v_2)^T,
    (\nabla\, v_3)^T
\right)^T
$, while for tensor fields $P$ with rows in ${\rm H}({\rm div}\,; \Omega)$, i.e.
$
P=\left(
    P_1^T,
    P_2^T,
    P_3^T
\right)$,  $P_i\in {\rm H}({\rm div}\,; \Omega)$
we define
 $ {\rm Div}\,P=\left(
   {\rm div}\, P_1,
    {\rm div}\,P_2,
    {\rm div}\,P_3
\right)^T$.
The corresponding Sobolev-spaces will be denoted by
$
  H^1(\Omega), H^1(\Div;\Omega).
 $

\section{The classical indeterminate couple stress model}\label{sectaxlb}\setcounter{equation}{0}

We are now shortly re-deriving the classical equations based on the $\nabla [\axl (\skw \nabla u)]$-formulation of the indeterminate couple stress model.  This part does not contain new results, see, e.g., \cite{MadeoGhibaNeffMunchKM} for further details, but is included for setting the stage for this contribution.

The linear isotropic indeterminate couple stress problem can be viewed as a minimization problem
\begin{align}
\int_\Omega\Big[ \mu\, \|\sym \nabla u\|^2+\frac{\lambda}{2}\, [\tr(\nabla u)]^2+W_{\rm curv}(\nabla \curl u)-\langle f, u\rangle \Big]dv \quad \rightarrow \text{min. w.r.t.} \quad  u,
\end{align}
subjected to geometric and mechanical boundary conditions, in part depending on the form of $W_{\rm curv}(\nabla \curl u)$, which will be specified later on.

In the following, in order to place the subject in the literature, we outline  some curvature energies proposed in different isotropic second gradient elasticity properly models:
\begin{itemize}
\item {\bf the indeterminate couple stress model} (Grioli-Koiter-Mindlin-Toupin model) \cite{Grioli60,Aero61,Koiter64,Mindlin62,Toupin64,Sokolowski72,grioli2003microstructures} in which the higher derivatives (apparently)  appear only through derivatives of the infinitesimal continuum rotation $\curl u$.  Hence, the curvature energy  has  the equivalent forms
\begin{align}\label{KMTe}
W_{\rm curv}( \nabla \curl\, u)&=\mu\, L_c^2\,\left[\frac{\alpha_1}{4}\, \|\sym \nabla \curl\, u\|^2+\frac{\alpha_2}{4}\,\| \skw \nabla \curl\, u\|^2\right]\notag\\
&=\mu\, L_c^2\,\left[{\alpha_1}\, \|\sym\nabla[\axl(\skw  \nabla u)]\|^2+{\alpha_2}\,\| \skw \nabla[\axl(\skw  \nabla u)]\|^2\right]\\
&=\mu\, L_c^2\,\left[\frac{\alpha_1}{4}\, \|\dev \sym \nabla \curl\, u\|^2+\frac{\alpha_2}{4}\,\| \skw \nabla \curl\, u\|^2\right].\notag
\end{align}
We  remark that the spherical part of the couple stress tensor remains {\bf indeterminate} since $\tr(\nabla \curl u)={\rm div} (\curl u)=0$. In order to prove the pointwise uniform positive definiteness it is assumed following \cite{Koiter64}, that ${\alpha_1}>0, {\alpha_2}>0$. Note that pointwise uniform positivity is often assumed \cite{Koiter64} when deriving analytical solutions for simple boundary value problems because it allows to invert the couple stress-curvature relation. It is clear  that pointwise positive definiteness is not necessary for well-posedness \cite{Neff_JeongMMS08}.
Mindlin \cite[p. 425]{Mindlin62} explained the relations between  Toupin's constitutive equations \cite{Toupin62} and  Grioli's \cite{Grioli60} constitutive equations and concluded that the obtained equations in the linearized theory are identical, since the extra constitutive parameter $\eta^\prime$ of Grioli's model does not explicitly appear in the equations of motion  but enters only the boundary conditions. The same extra constitutive coefficient appears in  Mindlin and Eshel's version \cite{Mindlin68}.

\item
  {\bf the modified - symmetric couple stress model - the conformal model}.  On the other hand, in the conformal case  \cite{Neff_Jeong_IJSS09,Neff_Paris_Maugin09} one may consider that $\alpha_2=0$, which makes the couple stress tensor $\widetilde{m}$ symmetric and trace free.  This conformal curvature case has been
considered  by Neff in  \cite{Neff_Jeong_IJSS09}, the curvature energy having the form
\begin{align}
W_{\rm curv}( \nabla \curl\, u)&=\mu\, L_c^2\,\frac{\alpha_1}{4}\, \|\sym \nabla \curl\, u\|^2=\mu\, L_c^2\,\alpha_1\, \|\dev \sym \nabla[\axl(\skw  \nabla u)\|^2.
\end{align}
Indeed, there are two major reasons uncovered  in \cite{Neff_Jeong_IJSS09} for using the modified couple stress model. First, in order to avoid singular stiffening behaviour for smaller and smaller samples in bending \cite{Neff_Jeong_bounded_stiffness09} one has to take $\alpha_2=0$. Second, based on a homogenization procedure invoking an intuitively appealing  natural ``micro-randomness" assumption (a strong statement of microstructural isotropy) requires conformal invariance, which is again equivalent to $\alpha_2=0$. Such a model is still well-posed \cite{Neff_JeongMMS08} leading to existence and uniqueness results with only one additional material length scale parameter, while it is {\bf not} pointwise uniformly positive definite.

\item {\bf the skew-symmetric couple stress model - the non-conformal model}.
  {Hadjesfandiari and Dargush} strongly advocate \cite{hadjesfandiari2011couple,hadjesfandiari2013fundamental,hadjesfandiari2013skew} the opposite extreme case, $\alpha_1=0$ and $\alpha_2>0$, i.e. they used the  curvature  energy
\begin{align}
W_{\rm curv}( \nabla \curl\, u)&=\mu\, L_c^2\,\frac{\alpha_2}{4}\, \|\skw \nabla({\rm curl}\, u)\|^2=\mu\, L_c^2\,\alpha_2\,  \|\axl \skw \nabla({\rm curl}\, u)\|^2=4\, \mu\, L_c^2\, \alpha_2\, \|{\rm curl}\,({\rm curl}\, u)\|^2\notag\notag.
\end{align}
In that model the nonlocal force stresses and the couple stresses are both assumed to be skew-symmetric. Their reasoning, based in fact on an incomplete understanding of boundary conditions is critically discussed in this paper  and generally refuted, while mathematically it is well-posed, which will also be shown.
 \end{itemize}

\subsection{Equilibrium and constitutive equations}
Taking free variations $\delta u\in C^\infty(\Omega)$ in the energy $W(\sym \nabla u,\nabla  \curl u)=W_{\rm lin}(\sym \nabla u)+W_{\rm curv}(\nabla  \curl u)$, where
\begin{align}\label{gradeq11}
W_{\rm lin}(\sym \nabla u)=&\mu\, \|\sym \nabla u\|^2+\frac{\lambda}{2}\, [\tr(\nabla u)]^2=\mu\, \|\dev \sym \nabla u\|^2+\frac{2\, \mu+3\,\lambda}{6}\, [\tr(\nabla u)]^2,\notag\\
W_{\rm curv}(\nabla  \curl u)=& \mu\,L_c^2\,[\alpha_1\, \|\dev\sym \nabla [\axl (\skw \nabla u)]\|^2+
\alpha_2\, \|\skw\nabla [\axl (\skw \nabla u)]\|^2],
\end{align}
 we obtain the virtual work principle
 \begin{align}\label{gradeq211}
\frac{\rm d}{\rm dt}\int_\Omega W(\nabla u+t\,\nabla \delta u)\,dv\Big|_{t=0}=&\int_\Omega\bigg[ 2\mu\,\langle\sym \nabla u, \sym \nabla \delta u \rangle+\lambda \tr(\nabla u)\,\tr( \nabla   \delta u)\notag\\&+\mu\,L_c^2\,[2\,  \alpha_1\, \langle \dev\sym \nabla [\axl (\skw \nabla u)],\dev\sym \nabla [\axl (\skw \nabla \delta u)]\rangle \\&+
2\, \alpha_2\, \langle \skw\nabla [\axl (\skw \nabla u)],\skw\nabla [\axl (\skw \nabla \delta u)]\rangle]+\langle f,\delta u\rangle \bigg]\, dv =0,\notag
\end{align}
where $f$ denotes the body force density.

The classical divergence theorem
 leads to
\begin{align}\label{germaneq311}
\int_\Omega\langle \Div (\sigma-\widetilde{\tau})+f, \delta u \rangle \, dv\underbrace{-\int_{\partial \Omega}\langle (\sigma-\widetilde{\tau}).\, n, \delta u\rangle \,dv
-\int_{\partial\Omega}\langle \widetilde{ {m}}.\, n, \axl (\skw \nabla \delta u) \rangle \, da}_{\tiny \text{the virtual power work of the surface forces}}=0,
\end{align}
where
 \begin{align}\label{consteq}
  \widetilde{\sigma}_{\rm total}&=\sigma-\widetilde{\tau}\not\in{\rm Sym}(3)\qquad\qquad\qquad\qquad\qquad\qquad\qquad\qquad\qquad\ \ \,  \text{total force-stress tensor} \vspace{1.2mm}\notag\\
 \sigma&=2\, \mu \, \sym \nabla u+\lambda \, \tr(\nabla u)\id \in {\rm Sym}(3)\qquad\qquad\qquad\qquad\quad\quad\ \,\text{local force-stress tensor} \vspace{1.2mm}\notag\\
  \widetilde{\tau}&=\dd\frac{1}{2}\anti
{\rm Div}[\widetilde{ {m}}]\in \so(3),\qquad\qquad\qquad\qquad\qquad\qquad\qquad\qquad \text{nonlocal force-stress tensor} \vspace{1.2mm}\\
\widetilde{ {m}}&=\mu\,L_c^2\,[{\alpha_1}\sym \nabla (\curl u)+ {\alpha_2}\,\skw\nabla (\curl u)]\qquad\qquad\qquad\quad \text{couple stress tensor}\vspace{1.2mm}\notag\\
&=\mu\,L_c^2\,[{\alpha_1}\dev\sym \nabla (\curl u)+ {\alpha_2}\,\skw\nabla (\curl u)]\text{}\vspace{1.2mm}\notag\\
&=\mu\,L_c^2\,[2\,{\alpha_1}\dev\sym \nabla[\axl (\skw \nabla u)]+2\, {\alpha_2}\,\skw\nabla [\axl (\skw \nabla u)]],\notag
\end{align}
and $n$ is the unit outward normal vector at the surface $\partial \Omega$.
 The  equilibrium equation are therefore
 \begin{align}\label{ec11}
 \Div \,\widetilde{\sigma}_{\rm total}+f=0.
 \end{align}
Note that  the local force-stress tensor $\sigma$ is always symmetric, the nonlocal force-stress  tensor $\widetilde{\tau}$  is automatically skew-symmetric, while the  second order hyperstress tensor (the couple stress tensor) $\widetilde{m}$  may or may not be symmetric, depending on the material parameters. The asymmetry of force stress is a hidden constitutive assumption, compare to \cite{NeffGhibaMadeoMunch}.

  \subsection{Boundary conditions}

  \subsubsection{Classical (incomplete) Grioli-Koiter-Mindlin-Tiersten boundary conditions}
Since  the variations at the boundary and the interior can be assigned independently, see \eqref{germaneq311}, we must have:
\begin{align}\label{integralasuprafata}
&-\int_{\partial \Omega}\langle (\sigma-\widetilde{\tau}).\, n, \delta u\rangle \,da
-\int_{\partial\Omega}\langle \widetilde{ {m}}.\, n, \axl(\skw \nabla \delta u) \rangle\, da=
0\qquad\text{or equivalently}\qquad\\
&-\int_{\partial \Omega}\langle (\sigma-\widetilde{\tau}).\, n, \delta u\rangle \,da
-2\,\int_{\partial\Omega}\langle \widetilde{ {m}}.\, n, \curl \delta u \rangle\, da=
0.\notag
\end{align}
This suggests 6 possible independent prescriptions of mechanical boundary conditions; three for the  normal components of the total force stress $(\sigma-\widetilde{\tau}).n$ and three for the normal components of the couple stress tensor.
 The possible Dirichlet boundary conditions on $\Gamma\subset \partial \Omega$  seem to be the 6 conditions\footnote{as indeed proposed by Grioli \cite{Grioli60} in concordance with the Cosserat kinematics for independent fields of displacements and microrotation.}
 \begin{align}\label{bc11}
 u=\widetilde{u}, \qquad \axl(\skw \nabla  u)=\widetilde{w} \quad (\text{or  equivalently} \ \ \curl u=2\,\widetilde{w}),
 \end{align}
 for two given functions $\widetilde{u}, \widetilde{w}:\mathbb{R}^3\rightarrow\mathbb{R}^3$ at the open subset  $\Gamma\subset\partial \Omega$ of the boundary (3+3 boundary conditions).

However, following Koiter  we note
\begin{remark}\label{remarkcc} {\rm [independent variations and curl] }Assume $u\in C^\infty(\overline{\Omega})$ and $u\Big|_{{\Gamma}}$ is known. Then ${\rm curl}\Big|_{{\Gamma}}$  exists and for any open subset ${\Gamma}\subset \partial \Omega$ the integral $\int_{{\Gamma}}\langle {\rm curl}
 \,u, n\rangle\, da$ is already known by Stokes theorem, while $\int_{{\Gamma}}\langle {\rm curl}\, u,\tau\rangle \,da$  is still free, where $\tau$ is any tangential vector field on  the open set ${\Gamma}\subset \partial \Omega$.
Only the two tangential components of ${\rm curl}\, u$  may be independently prescribed on an open subset of the boundary.
\end{remark}

Already, Mindlin and Tiersten \cite{Mindlin62} have rightly  remarked that also in this formulation  only 5 mechanical boundary conditions  can  be prescribed. They rewrote \eqref{integralasuprafata} in a further separated form
\begin{align}\label{integralasuprafata0}
-\int_{\partial \Omega}\langle (\sigma-\widetilde{\tau}).\, n -\underbrace{\frac{1}{2} n\times \nabla[\langle (\sym \widetilde{ {m}}).n,n\rangle]}_{\begin{array}{c}\text{\tiny{performs already}}\vspace{-1.5mm}\\ \text{\tiny{work only against \ $\delta u$} }\end{array}}, \delta u\rangle \,da
-\int_{\partial\Omega}\langle (\id-n\otimes n)\,\widetilde{ {m}}.n,\!\!\!\!\!\!\!\!\underbrace{(\id-n\otimes n)\,[\axl(\skw \nabla \delta u)]}_{\begin{array}{c}\text{\tiny{cannot be assigned arbitrarily independent of  $\delta u$}}\vspace{-1.5mm}\\
\text{\tiny{since it still contains certain}}\vspace{-1.5mm}\\\text{\tiny{ "tangential" derivatives of $ \delta u\in C^\infty(\Omega)$}}
\end{array} }\!\!\!\!\!\!\!\!\!\!\!\rangle\,
 da\notag\\
 \\
 =-\int_{\partial \Omega}\langle (\sigma-\widetilde{\tau}).\, n -\underbrace{\frac{1}{2} n\times \nabla[\langle (\sym \widetilde{ {m}}).n,n\rangle]}_{\begin{array}{c}\text{\tiny{performs already}}\vspace{-1.5mm}\\ \text{\tiny{work only against \ $\delta u$} }\end{array}}, \delta u\rangle \,da
-2\,\int_{\partial\Omega}\langle (\id-n\otimes n)\,\widetilde{ {m}}.n,(\id-n\otimes n)\,[\curl \delta u]\rangle\,
 da=
0,\notag
\end{align}
 which already shows that the second term in \eqref{integralasuprafata} still contains contributions which perform work against $\delta u$, namely
 $\frac{1}{2} n\times \nabla[\langle (\sym \widetilde{ {m}}).n,n\rangle]$, while the remaining higher order term $(\id-n\otimes n)\,\widetilde{ {m}}.n$  performs work against combinations of second derivatives. Mindlin and Tiersten \cite{Mindlin62}  concluded that 3 boundary conditions derive from the first integral (correctly) and two other  from the second integral, since \cite[p.~432]{Mindlin62} ``the normal component of the couple stress vector [$\langle \widetilde{ {m}}.n,n\rangle=\langle \sym\widetilde{ {m}}.n,n\rangle$] on" $\partial \Omega$ ``enters only in the combination with the force-stress vector shown in the coefficient of" $\delta u$ ``in the surface integral (our first term on the right hand side of \eqref{integralasuprafata0}).

Therefore, Mindlin and Tiersten \cite{Mindlin62} concluded that the boundary conditions consists in:
\begin{itemize}
\item \textbf{Geometric boundary conditions} on $\Gamma\subset \partial \Omega$:
\begin{align}\label{bcme1}
 &\hspace{4.52cm}u\ =\ \widetilde{u}^0, \qquad \qquad \quad\qquad\quad\qquad\qquad \qquad   \qquad (3\  \text{bc})\notag\\
 &\hspace{0.06cm}\left\{
 \begin{array}{rcl}
(\id -n\otimes n).\,\axl(\skw \nabla  u)&=&\ \ \!\!\!(\id -n\otimes n).\,\axl(\skw \nabla  \widetilde{u}^0),\vspace{1.2mm}\\
\text{or}\qquad  (\id -n\otimes n).\,\curl  u&=&\!\!\!2\, (\id -n\otimes n).\,\curl  \widetilde{u}^0, \qquad \qquad \qquad   \qquad (2\  \text{bc})
\end{array}\right.
 \end{align}
  for  a given function $\widetilde{u}^0:\mathbb{R}^3\rightarrow\mathbb{R}^3$ at the boundary.
 The latter condition prescribes only the tangential component of $\axl(\skew \nabla u)$. Therefore, we may  prescribe only 3+2  independent boundary conditions.
\item  \textbf{Traction boundary conditions} on $\partial \Omega\setminus\overline{\Gamma}$:
\begin{align}\label{bcme2}
(\sigma-\widetilde{\tau}).\, n-\frac{1}{2} n\times \nabla[\langle (\sym \widetilde{ {m}}).n,n\rangle]&=\widetilde{t},\qquad\qquad \qquad\qquad\text{traction}\qquad\qquad \qquad\qquad\quad\text{ \ (3 bc)}\notag\\
(\id-n\otimes n)\,\widetilde{ {m}}.n&=(\id-n\otimes n)\,\widetilde{g}, \qquad \text{``double force normal traction" (2 bc)}
\end{align}
 for prescribed functions $\widetilde{t},\widetilde{g}:\mathbb{R}^3\rightarrow\mathbb{R}^3$ at the boundary.
\end{itemize}
However, while \eqref{bcme1} and \eqref{bcme2} correctly describe the maximal number of independent boundary conditions in the indeterminate couple stress model and these conditions have been rederived again and again by Yang et al. \cite{Yang02}, Park and Gao \cite{Park07}, \cite{lubarda2003effects}, etc. among others they are  not correct. This is explained in the following paragraphs.

\subsubsection{Complete (corrected) traction boundary conditions}\label{sectioncomplete}

 The indeterminate couple stress model is not simply obtained as a constraint Cosserat model \cite{Sokolowski72}, i.e. assuming that $\overline{A}=\axl(\skew \nabla u)$. In the indeterminate couple stress model the only independent kinematical degree of freedom is $u$. We understand that the indeterminate couple stress model constructed as a constraint Cosserat model represents at most an approximation of the  indeterminate couple stress model, in the sense that the boundary conditions are not correctly/completely considered. We remark that the quantity $\langle \widetilde{ {m}}.n,(\id-n\otimes n)[\axl(\skw \nabla \delta u)]\rangle$ does still contain contributions performing work against $\delta u$ alone (even though there is a projection $\id -n\otimes n$ involved), which can be assigned arbitrarily and are therefore somehow related to   independent variation  $\delta u$. This case is not similar to the Cosserat theory in which we assume a priori that  displacement $u$ and microrotation $\overline{A}\in\so(3)$ are independent kinematical degrees of freedom.

At this point, it must also be considered that the tangential trace of
the gradient of virtual displacement can be integrated by parts once
again and that the surface divergence theorem can be applied to this
tangential  part of $\nabla\delta{u}$. Let  ${n}$ be the unit normal vector
at a considered
surface point.
As it is well known from differential geometry, the projectors $\id -{n}\otimes{n}$ and ${n}\otimes{n}$ allow to split a given vector or tensor field in one part projected
on the plane tangent to the considered surface and one projected on
the normal to such surface (see also \cite{NthGrad,TheseSeppecher,dell2012beyond} for details)). For the sake of simplicity we  assume that $\partial\Omega$ is a smooth  surface of class $C^2$. As in the notation section, we consider the curve $\partial \Gamma$ which joins the open subsets $\Gamma$ and $\partial \Omega\setminus \Gamma$ of the boundary. Therefore, in our case and in our abbreviations, the surface divergence \cite[p.~58, ex.~7]{gurtin2010mechanics} reads:
\begin{equation}
\int_{\partial \Omega}\mathrm{Div}^{S}\left(v\right)da:=\int_{\partial \Omega}\langle \id-n\otimes n,\nabla\left((\id-n\otimes n)\cdot v\right)]\rangle\, da=\int_{\partial \Gamma}\jump{\:\langle v,{\nu}\rangle\text{ }}\,ds.\label{eq:Surface_Div_Th}
\end{equation}
for any field $v\in \mathbb{R}^3$.

 Indeed, using the surface divergence theorem, we have obtained in \cite{MadeoGhibaNeffMunchKM} that
\begin{align}\label{axldeltan}
-\int_{\partial \Omega}\langle (\sigma&-\widetilde{\tau}).\, n, \delta u\rangle \,da
-\int_{\partial \Omega}\langle \widetilde{ {m}}.\, n, \axl(\skw \nabla \delta u) \rangle\, da\notag\\
 =&
 -\int_{\partial \Omega}\langle (\sigma-\frac{1}{2}\anti
{\rm Div}[\widetilde{ {m}}]).\, n -\frac{1}{2} n\times \nabla[\langle (\sym \widetilde{ {m}}).n,n\rangle]\\
&\qquad\qquad\ -\frac{1}{2}\{\nabla[(\anti[(\id-n\otimes n)\widetilde{ {m}}.\, n])\, (\id -{n}\otimes{n})]:(\id -{n}\otimes{n}), \delta u\rangle \,da\notag\\& -\frac{1}{2}
\int_{\partial \Omega}\underbrace{\langle (\id -n\otimes n)\anti[(\id-n\otimes n)\widetilde{ {m}}.\, n].n,  \nabla \delta u.n \rangle}_{\tiny\begin{array}{c}\text{\tiny{completely}}\ {\tiny \delta u}\text{\tiny{-independent second}}\vspace{0mm}\\
\text{\tiny{order normal variation of the gradient}}
\end{array} }  da-\frac{1}{2}
\int_{\partial \Gamma}\langle \jump{\anti[(\id-n\otimes n)\widetilde{ {m}}.\, n]. \,\nu}, \delta u\rangle ds.\notag
\end{align}
Hence, there are indeed two terms
$$(\sigma-\frac{1}{2}\anti
{\rm Div}[\widetilde{ {m}}]).\, n -\frac{1}{2} n\times \nabla[\langle (\sym \widetilde{ {m}}).n,n\rangle]-\frac{1}{2}\nabla[(\anti[(\id-n\otimes n)\widetilde{ {m}}.\, n])(\id -{n}\otimes{n})]:(\id -{n}\otimes{n})$$ and $$\jump{\anti[(\id-n\otimes n)\widetilde{ {m}}.\, n]. \,\nu}$$
 which perform work against $\delta u$, while only the term
 $$(\id -n\otimes n)\anti[(\id-n\otimes n)\widetilde{ {m}}.\, n].n$$
  is  related solely   to the independent second order normal variation of the gradient  $ \nabla \delta u.n$.
This split of the boundary condition is not the one as obtained e.g. by Gao and Park \cite{park2008variational} and seems to be entirely new in the context of the indeterminate couple stress model.

Therefore, we need to adjoin on $\partial \Omega$ the following complete set of boundary conditions:
\begin{itemize}
\item Geometric (essential) boundary conditions on $\Gamma\subset\partial \Omega$:
 \begin{align}\label{bc1110}
u&=\widetilde{u}^0, \qquad \qquad \quad\quad\quad\quad\quad \quad\ \   \qquad\qquad\qquad\qquad (3\  \text{bc})\\
(\id-n\otimes n)(\nabla u).n&= (\id-n\otimes n)(\nabla \widetilde{u}^0).n,  \qquad \quad  \qquad\qquad\qquad\qquad (2\  \text{bc})\notag
 \end{align}
 where $\widetilde{u}^0:\mathbb{R}^{3}\rightarrow\mathbb{R}^{3}$ is a prescribed function (i.e. 3+2=5 boundary conditions), as in \cite{Mindlin62}.
\item Traction boundary conditions on $\partial \Omega\setminus\overline{\Gamma}$:
  \begin{align}\label{bc1001}\hspace{-1.3cm}
        \begin{array}{rcl}
(\sigma-\widetilde{\tau}).\, n -\frac{1}{2} n\times \nabla[\langle (\sym \widetilde{ {m}}).n,n\rangle] \hspace{4cm}&&\vspace{1.2mm}\\
-\frac{1}{2}\{\nabla[(\anti[(\id-n\otimes n)\widetilde{ {m}}.\, n])(\id-n\otimes n)]:(\id-n\otimes n) &=&t,\vspace{1.2mm}\\
  \dd(\id -n\otimes n)\anti[(\id-n\otimes n)\widetilde{ {m}}.\, n].n&=&(\id-n\otimes n)\,g
 \end{array}
&
 \begin{array}{r}
\\(3\  \text{bc})\vspace{1.2mm}\\
(2\  \text{bc})
 \end{array}
 \end{align}
  where  $t,  g:\mathbb{R}^{3}\rightarrow\mathbb{R}^{3}$ are prescribed functions on $\partial \Omega\setminus \overline{\Gamma}$.
 \item  Jump boundary conditions on $\partial \Gamma\subset \partial \Omega$:
 \begin{align}\begin{array}{rcl}
\jump{\anti[(\id-n\otimes n)\widetilde{ {m}}.\, n]. \,\nu}&=&\widetilde{\pi},\hspace{2.3cm}
 \end{array}
\qquad
  \qquad \qquad\qquad\qquad \begin{array}{r}
(3\  \text{bc})
 \end{array}
 \end{align}
 where $\widetilde{\pi}$ is prescribed on $\partial \Gamma$ and leads to 3 boundary conditions.
\end{itemize}

\section{The Hadjesfandiari and Dargush's postulate}\setcounter{equation}{0}
Let us now turn to  Hadjesfandiari and Dargush's far reaching claims.
 In the abstract of their  paper \cite{hadjesfandiari2011couple} Hadjesfandiari and Dargush write:
 \begin{quote}
 \textit{``By relying on the definition of admissible boundary conditions, the principle of virtual work and some
kinematical considerations, we establish the skew-symmetric character of the couple-stress tensor {\rm [}$\widetilde{m}${\rm ]} in
size-dependent continuum representations of matter. This fundamental result, which is independent of
the material behavior {\rm [e.g. isotropy]}, resolves all difficulties in developing a consistent couple stress theory."}
 \end{quote}
 \noindent
In their appendix of \cite[p. 1263-1264]{hadjesfandiari2013fundamental} they add:
\begin{quote}
\textit{``Therefore, the present determinate theory is not mathematically a special case of {\rm [the]} indeterminate {\rm [couple stress]} theory obtained by letting  {\rm [}$\alpha_1=0${\rm ]}. This is just a coincidence for the linear isotropic case where equations in both theories have some similarities. It should be realized that the determinate theory is not simply about fixing the constitutive equations for {\rm [the]} linear isotropic couple stress theory of elasticity.  The present size-dependent couple stress theory is the consistent couple stress theory in continuum mechanics. This has been achieved by discovering the skew-symmetric character of the couple-stress tensor {\rm [}$\widetilde{m}${\rm ]}. Mindlin, Tiersten and Koiter did not recognize the mean curvature tensor {\rm [}$\skw[\nabla{\rm curl}\, u]${\rm ]} as the consistent measure of deformation in continuum mechanics."}\footnote{In \cite[p. 1283, A.16]{hadjesfandiari2013fundamental} Hadjesfandiari and Dargush erroneously  take the strict positivity of curvature parameters $\alpha_1,\alpha_2>0$, used initially by Koiter, Mindlin and others, as belonging to the definition of the indeterminate couple stress model. This is clearly not the case. It can be shown that both limit cases $\alpha_1>0, \alpha_2=0$ (modified couple stress theory) and $\alpha_1=0, \alpha_2>0$ (Hadjesfandiari and Dargush choice) are mathematically admitted with the provision of using the attendant correct boundary conditions being defined by the solution space depending on the values of $\alpha_1, \alpha_2$, see Section  \ref{existHD}.}
\end{quote}

In our understanding this claim is completely unfounded. Assuming isotropic response, we treat their ``consistent and determinate" couple stress theory as the linear and isotropic accepted indeterminate couple stress model with constitutive parameters $\alpha_1=0$, $\alpha_2>0$  in \eqref{consteq}.

Turning to their most important claim regarding the skew-symmetry of the coule stress tensor $\widetilde{m}$ we will exhibit their line of thought. Their reasoning is based on their fundamental hypothesis that the normal component of the couple stress traction vector $\langle \widetilde{m}.n,n\rangle $ should vanish on any bounding surface of an arbitrary volume\footnote{They also wrote \cite[p.13]{hadjesfandiari2014evo}:
 \textit{``...the corresponding generalized force must be zero and, for the normal component of the surface moment-traction vector {\rm [}$\widetilde{m}.n${\rm ]}, we must enforce the condition {\rm [}$\langle \widetilde{m}.n,n\rangle=0${\rm ]}."}.}.  In the case $\alpha_1=0$ we have $\widetilde{m}\in\so(3)$ and $\langle \widetilde{m}.n,n\rangle =0$ is satisfied automatically.

\smallskip

\noindent In our notation the argument  of Hadjesfandiari and  Dargush is given as follows \cite[p. 12-13]{hadjesfandiari2010polar}:
\begin{quote}
\textit{``From kinematics, since {\rm [$\omega^{nn}:=\langle {\rm curl} u,n\rangle$]} is not an independent generalized degree of freedom, its apparent
corresponding generalized force must be zero. Thus, for the normal component of the surface
couple vector {\rm [$\widetilde{m}.n$]}, we must enforce the condition
\[
[\langle \widetilde{m}.n,n\rangle=0].
\]
{\rm [...]} we notice that
the energy equation can be written for any arbitrary volume with arbitrary surface within the
body. Therefore, for any point on any arbitrary surface with unit normal $n$, we must have
\[
[\langle \widetilde{m}.n,n\rangle=0].
\]
Since $n_in_j$ is symmetric and arbitrary in {\rm [..]}, {\rm [$\widetilde{m}$]}  must be skew-symmetric. Thus,
{\rm [$\widetilde{m}^T=-\widetilde{m}$.]}
This is the fundamental property of the couple-stress tensor in polar continuum mechanics,
which has not been recognized previously."} (see also \cite[p. 2500]{hadjesfandiari2011couple}).
\end{quote}

\noindent Let us interpret  this statement. We rephrase it for our purpose through  the formulation of an implicit
\begin{quote}
\vspace{-0.5cm}
\begin{postulate}{\rm [Hadjesfandiari and Dargush as we understand it]}\label{postulate} In any extended continuum model only the total force stress traction vector $(\sigma-\widetilde{\tau}).\, n$ should perform work against the independent virtual displacement $\delta u$ at the part $\partial \Omega\setminus{\overline{\Gamma}}$ of the boundary $\partial \Omega$, where traction boundary conditions are applied.
\end{postulate}
\end{quote}

\begin{remark}
 Incidentally, this postulate is automatically satisfied in classical elasticity, Cosserat and micromorphic models \cite{Neff_Forest_jel05}, since the possible variations of the field variables are independent anyway. Whether such a postulate can be satisfied in a higher gradient continuum is the concern of Hadjesfandiari and Dargush. We will see that this is not possible.
 \end{remark}
Hadjesfandiari and Dargush apply this postulate to the classical (incomplete) Mindlin and Tiersten's format of the boundary conditions, namely   \eqref{bcme1} and \eqref{bcme2}. Inspection of the indeterminate couple stress model within the framework of these  erroneous classical boundary conditions, see e.g. \eqref{inspHD} and \eqref{integralasuprafata01}, shows that choosing $\sym \widetilde{m}=0$ is indeed  sufficient for this postulate to be satisfied.

Hadjesfandiari and Dargush accept
 \begin{align}
\sigma.n=\sigma_n,\qquad \widetilde{m}.n=m_n,
 \end{align}
 and they also realize correctly that the number of geometric or mechanic boundary conditions is 5 since the tangential component of the test function $\delta u$ cannot independently be varied, which is by now well established. To this aim, they spilt the term
 \begin{align}\label{HDpw}
 \int_{\partial V}\langle \widetilde{m}.n, \axl\skew \nabla \delta u\rangle \,da&=2\,\int_{\partial V}\langle \widetilde{m}.n, \curl \delta u\rangle\, da\\
 &=2\,\int_{\partial V}\underbrace{\langle (n\otimes n)\, \widetilde{m}.n}_{\text{\tiny normal part}}, \curl \delta u\rangle\, da+2\,\int_{\partial V}\underbrace{\langle (\id-n\otimes n)\,\widetilde{m}.n}_{\text{\tiny tangential part}}, \curl \delta u\rangle\, da\notag
 \\
 &=2\,\int_{\partial V}\langle\widetilde{m}.n, \underbrace{(n\otimes n).\,\curl \delta u\rangle}_{\text{\tiny normal part}}\, da+2\,\int_{\partial V}\langle\widetilde{m}.n,\underbrace{(\id-n\otimes n).\,\curl \delta u\rangle}_{\text{\tiny tangential part}}\, da\notag
 \end{align}
 on any arbitrary subdomain $V\subset \Omega$, into its tangential and normal part. They  observe that the normal part
  \begin{align}\label{HDPw2}
 \int_{\partial V}\langle (n\otimes n)\, \widetilde{m}.n, \curl \delta u\rangle\, da=\int_{\partial V}\underbrace{\langle (\sym \widetilde{ {m}}).n,n\rangle}_{\text{\tiny generalized force}} \,  \!\!\!\underbrace{\langle \curl \delta u, n\rangle}_{\begin{array}{c}\vspace{-6mm}\\\text{\tiny no independent }\vspace{-2mm}\\
 \text{\tiny degree of freedom $(\surd)$}\end{array}}\, da
 \end{align}
 cannot be prescribed independently of $\delta u$, since indeed their $\omega^{nn}:=\langle \curl \delta u, n\rangle$  cannot  independently be prescribed due to Stokes theorem. To avoid a somehow felt inconsistency, they state that the corresponding generalized force must be zero and, for the normal component of the surface moment-traction vector, they enforce accordingly the (misguided)  condition
 \begin{align}
 \underbrace{\langle (\sym \widetilde{ {m}}).n,n\rangle}_{\text{\tiny generalized force}} =0
 \end{align}
 on any arbitrary subdomain $V\subset \Omega$ having the boundary $\partial V$.

The equilibrium equations considered by Hadjesfandiari and Dargush ($\alpha_1=0$)\cite{hadjesfandiari2011couple} read therefore
\begin{align}\label{HD1}
 \Div \widetilde{\sigma} _{\rm total}+f=0,
 \end{align}
 where the total force stress is given by
 \begin{align}\label{HD2}
  \widetilde{\sigma}_{\rm total}&=\sigma-\widetilde{\tau}\not\in{\rm Sym}(3)\qquad\qquad\qquad\qquad\qquad\qquad\qquad\qquad\qquad\qquad\qquad\ \ \,  \text{total force-stress tensor} \vspace{1.2mm}\notag\\
 \sigma&=2\, \mu \, \sym \nabla u+\lambda \, \tr(\nabla u)\id \in {\rm Sym}(3)\qquad\qquad\qquad\qquad\quad\quad\qquad\qquad\ \,\text{local force-stress tensor} \vspace{1.2mm}\notag\\
  \widetilde{\tau}&=\dd\frac{1}{2}\anti
{\rm Div}[\widetilde{ {m}}]\in \so(3),\qquad\qquad\qquad\qquad\qquad\qquad\qquad\qquad\qquad\qquad \text{nonlocal force-stress tensor} \vspace{1.2mm}\\
\widetilde{ {m}}&=\mu\,L_c^2\, {\alpha_2}\,\skw\nabla (\curl u)=2\,\mu\,L_c^2\, {\alpha_2}\,\skw\nabla [\axl (\skw \nabla u)]]\in \so(3)\quad\  \text{couple stress tensor}.\notag
\end{align}

 To the  equilibrium equation, Hadjesfandiari-Dargush  adjoin on $\partial \Omega$  the following boundary conditions
\begin{itemize}
\item \textbf{Geometric boundary conditions} on $\Gamma\subset \partial \Omega$:
\begin{align}\label{HD3}
 &\hspace{4.45cm}u\ =\ \widetilde{u}^0, \notag\\
 &\hspace{0.06cm}\left\{
 \begin{array}{rcl}
(\id -n\otimes n).\axl(\skw \nabla  u)&=&\ \ \!\!\!(\id -n\otimes n).\axl(\skw \nabla  \widetilde{u}^0),\\
\text{or}\qquad  (\id -n\otimes n).\,\curl  u&=&\!\!\!2\, (\id -n\otimes n).\,\curl  \widetilde{u}^0,
\end{array}\right.
 \end{align}
 for  a given function $\widetilde{u}^0:\mathbb{R}^3\rightarrow\mathbb{R}^3$ at the boundary.
 The latter condition prescribes only the tangential component of $\axl(\skew \nabla u)$. Therefore, they prescribe(correctly) only 3+2  independent boundary conditions.

 We may consider the equivalent geometric boundary conditions
 \begin{align}\label{HD30}
 \begin{array}{rclll}
 u&=&\widetilde{u}_0 &\text{on}&\Gamma,\\
 (\id-n\otimes n)\,\nabla u.\,n&=& (\id-n\otimes n)\,\nabla \widetilde{u}_0.\,n &\text{on} &\Gamma,
 \end{array}
 \end{align}
 where $\widetilde{u}_0:\mathbb{R}^3\rightarrow\mathbb{R}^3$  is given, i.e. 3+2 boundary conditions.
\item  \textbf{Traction boundary conditions} on $\partial \Omega\setminus\overline{\Gamma}$:
\begin{align}
\text{The Hadjesfandiari-Dargush-choice:}\qquad
\begin{array}{rclll}\label{HBBC}
 (\sigma+\widetilde{\tau}).\, n&=&\widetilde{t}& \text{on}& \partial \Omega\setminus\overline{\Gamma},\\
 (\id-n\otimes n)\widetilde{ {m}}.n&=&(\id-n\otimes n).\, \widetilde{h} & \text{on}&\partial \Omega\setminus\overline{\Gamma},
 \end{array}
\end{align}
where $\widetilde{t},\widetilde{h}:\mathbb{R}^3\rightarrow\mathbb{R}^3$ are given.
\end{itemize}

\section{No reason for the Hadjesfandiari-Dargush formulation with skew-symmetric couple stress tensor}\label{HDF}\setcounter{equation}{0}

Let us try to follow the argument of Hadjesfandiari and Dargush (bona fide).  It follows that the split considered by
 Hadjesfandiari-Dargush is not complete, since they nowhere do use the surface divergence theorem and they do not prescribe fully independent geometrically boundary conditions. We see, on the contrary, that, when integrated over $\partial \Omega$, even upon the Hadjesfandiari and Dargush restriction the  remaining tangential part from \eqref{HDpw}, namely
\begin{align}
 \int_{\partial \Omega}\underbrace{\langle (\id-n\otimes n)\widetilde{m}.n, \curl \delta u\rangle}_{\text{\tiny tangential part}} da=&\quad\frac{1}{4} \int_{\partial \Omega}\Big\{\langle \nabla[(\anti(\widetilde{ {m}}.\, n))\,(\id-n\otimes n)]: (\id-n\otimes n), \delta u\rangle \\
 &\qquad\qquad\qquad+
\underbrace{\langle (\id -n\otimes n)\anti(\widetilde{ {m}}.\, n).n,  \nabla \delta u.n \rangle}_{\begin{array}{c}\vspace{-6mm}\\\text{\tiny completely} \,\mbox{\tiny $\delta u$}\text{\tiny-independent second order}\vspace{-2mm}\\
\text{\tiny normal variation of gradient}
\end{array} }\Big\} \,da\notag,
 \end{align}
still contains some parts with independent degrees of freedom, performing work against the normal derivative $\nabla \delta u. n$.

Further on, we explain this in more detail. Looking at the Mindlin and Tiersten approach, see also \cite{MadeoGhibaNeffMunchKM},  in the framework of Hadjesfandiari-Dargush's assumption $\widetilde{m}\in\so(3)$, we get
\begin{align}\label{inspHD}
\langle\widetilde{ {m}}.\,n, \axl(\skw \nabla \delta u)\rangle&=
\langle (\id-n\otimes n)\,\widetilde{ {m}}.n,\axl(\skw \nabla \delta u)\rangle\notag\\&\quad+\frac{1}{2}\langle n,\curl[\!\!\!\!\!\!\!\!\!\!\!\!\underbrace{\langle (\sym \widetilde{ {m}}).n,n\rangle}_{\begin{array}{c}\vspace{-6mm}\\ \text{\tiny zero generalized force}\vspace{-2mm}\\
\text{\tiny according to}\vspace{-2mm}\\\text{\tiny Hadjesfandiari-Dargush's assumption}\end{array}}\!\!\!\!\!\!\!\!\!\!\!\!\!\!\!\!\, \delta u]\rangle
-\frac{1}{2}\underbrace{\langle n\times \nabla[\langle (\sym \widetilde{ {m}}).n,n\rangle]}_{\begin{array}{c}\vspace{-6mm}\\
\text{\tiny performs work against} \ \mbox{\tiny $\delta u$}\vspace{-2mm}\\
\text{\tiny and should therefore vanish}\vspace{-2mm}\\
\text{\tiny corresponding to their Postulate \ref{postulate}}
\end{array}}, \delta u\rangle,\\
&=
\langle (\id-n\otimes n)\,\widetilde{ {m}}.n,\axl(\skw \nabla \delta u)\rangle=\frac{1}{2}
\langle (\id-n\otimes n)\,\widetilde{ {m}}.n,\curl \delta u\rangle.\notag
\end{align}
Therefore \eqref{integralasuprafata0} can immediately be written as
\begin{align}\label{integralasuprafata01}
-\int_{\partial \Omega}\langle \Big\{(\sigma-\widetilde{\tau}).\, n, \delta u\rangle\, da
-\int_{\partial\Omega}\langle(\id-n\otimes n)\, \widetilde{ {m}}.n,(\id-n\otimes n)\,[\axl(\skw \nabla \delta u)]\rangle\,
 da=
0,
\end{align}
which is used and accepted by Hadjesfandiari and Dargush, see \eqref{HDpw}. We comprehend their curvature parameter choice (based on the incomplete format of the independent boundary conditions): normal tractions would be automatically completely separated into pure total force-stress tractions and pure couple stress tractions.

 However, after integration, the second term from \eqref{integralasuprafata01} $\langle(\id-n\otimes n)\, \widetilde{ {m}}.n,(\id-n\otimes n)\,[\axl(\skw \nabla \delta u)]\rangle$ leads to two new quantities: one performing work against the normal derivative $\nabla \delta u. n$ and one still performing work against  $\delta u$. This means that the assumption of Hadjesfandiari and Dargush does not remove all quantities which may also perform work against $\delta u$, besides  the total force stress tensor $(\sigma-\widetilde{\tau}).n$. This fact follows as a direct consequence of \eqref{axldeltan}, since, also in the Hadjesfandiari-Dargush formulation similar to \cite{MadeoGhibaNeffMunchKM,MadeoGhibaNeffMunchCRM}, it is  possible for us to use the surface divergence theorem  and   to obtain \begin{align}
-\int_{\partial \Omega}\langle (\sigma&-\widetilde{\tau}).\, n, \delta u\rangle \,da
-\int_{\partial\Omega}\langle \widetilde{ {m}}.\, n, \axl(\skw \nabla \delta u) \rangle\, da\notag\\
 =&
 -\int_{\partial \Omega}\langle (\sigma-\widetilde{\tau}).\, n -\frac{1}{2}\nabla[(\anti[(\id-n\otimes n)\widetilde{ {m}}.\, n])(\id-n\otimes n)]:(\id-n\otimes n), \delta u\rangle \,da\notag\\& -\frac{1}{2}
\int_{\partial \Omega}\langle (\id -n\otimes n)\anti[(\id-n\otimes n)\widetilde{ {m}}.\, n].n,  \nabla \delta u.n \rangle \,da\\&-\frac{1}{2}
\int_{\partial\Gamma}\langle  \jump{\anti((\id-n\otimes n)\widetilde{ {m}}.\, n)}.\, \nu, \delta u\rangle\, ds\notag,\notag
\\
 =&
 -\int_{\partial \Omega}\langle (\sigma-\widetilde{\tau}).\, n -\frac{1}{2}\nabla[(\anti[\widetilde{ {m}}.\, n])(\id-n\otimes n)]:(\id-n\otimes n), \delta u\rangle \,da\notag\\& -\frac{1}{2}
\int_{\partial \Omega}\langle (\id -n\otimes n)\anti[\widetilde{ {m}}.\, n].n,  \nabla \delta u.n \rangle \,da\notag-\frac{1}{2}
\int_{\partial\Gamma}\langle  \jump{\anti(\widetilde{ {m}}.\, n)}.\, \nu, \delta u\rangle \,ds\notag,\notag
\end{align}
for all variations $ \delta u\in C^\infty(\Omega)$, where we have used that $\langle \widetilde{m}.\, n,n\rangle=0$ implies
 $$(\id-n\otimes n)\widetilde{ {m}}.\, n=\widetilde{ {m}}.\, n-n\, \langle \widetilde{m}.\, n,n\rangle=\widetilde{ {m}}.\, n.$$
  We consider this last representation as the only correct form following the Hadjesfandiari-Dargush assumption $\langle \widetilde{m}.\, n,n\rangle=0$ made in \eqref{HDPw2}.

Similar to the correction to the  Mindlin and Tiersten's approach with the specification that $\widetilde{m}\in \so(3)$ (which is now seen to be not necessary but possible), we arrive rather, upon the Hadjesfandiari-Dargush assumption $\langle \widetilde{m}.\, n,n\rangle=0$, at the traction boundary condition
 \begin{align}\label{HBBC1}
 \hspace{-0.5cm}
        \begin{array}{rcl}
 [(\sigma-\widetilde{\tau}).\, n -\frac{1}{2}\nabla[(\anti[(\id-n\otimes n)\widetilde{ {m}}.\, n])](\id -{n}\otimes{n})]:(\id -{n}\otimes{n})\,(x) &=&\widetilde{g}(x),\vspace{1.2mm}\\
  \dd[(\id -n\otimes n)\anti[(\id-n\otimes n)\widetilde{ {m}}.\, n].n]\,(x)&=&[(\id-n\otimes n)\,\widetilde{s}]\,(x)
 \end{array}
 \
 \begin{array}{r}
(3\  \text{bc})\vspace{1.2mm}\\
(2\  \text{bc})
 \end{array}
 \end{align}
 on $\partial \Omega\setminus {\Gamma}$, while on $\partial \Gamma$ we have to prescribe  the jump conditions
 \begin{align}
 \hspace{4.5cm}
 \begin{array}{rcl}
\dd\{\jump{\anti[(\id-n\otimes n)\widetilde{ {m}}.\, n]}.\,\nu\}(x)&=&\widetilde{\pi}(x),
 \end{array}
 \quad
\quad\quad\qquad\
 \begin{array}{r}
(3\  \text{bc})
 \end{array}\notag
 \end{align}

From \eqref{HBBC1},  we see finally that their intended  response is not satisfied:
\begin{remark}
Assuming that $\langle (\sym \widetilde{ {m}}).n,n\rangle=0$ (the Hadjesfandiari-Dargush assumption) implies only that
\begin{align}
\int_{\partial \Omega}\langle \nabla[(\anti[(n\otimes n)\widetilde{ {m}}.\, n])(\id -{n}\otimes{n})]:(\id -{n}\otimes{n}), \delta u\rangle \,da=0.
\end{align}
However, from $$\nabla[(\anti[\widetilde{ {m}}.\, n])(\id -{n}\otimes{n})]:(\id -{n}\otimes{n}),$$ there remains another  contribution $$\nabla[(\anti[(\id-n\otimes n)\widetilde{ {m}}.\, n])(\id -{n}\otimes{n})]:(\id -{n}\otimes{n})$$ which still performs  work against  $\delta u$. The raison d'\^{e}tre of the Hadjesfandiari-Dargush formulation, as we understand it, was to avoid that any parts related to ``couple stress normal traction" $\widetilde{m}.n$, other than the contributions to the total force-stresses, i.e. other than $\sigma-\widetilde{\tau}$,  would perform work against $\delta u$. Our calculation shows that this requirement can never be satisfied in any higher gradient elasticity  theory.
\end{remark}

Continuing, Hadjesfandiari tries to support the claim  of a skew-symmetric couple stress tensor with some independent motivations in \cite{hadjesfandiari2013skew}. There, he defines the  ``torsion-tensor" $\Chi$ and the ``mean curvature tensor" $\omega$
 \begin{align}
 \Chi(u)=\sym \nabla \curl u, \qquad \omega(u)=\skew\, \nabla \curl u,
 \end{align}
 respectively.  In Section 3.2. of
  \cite[eq.~32]{hadjesfandiari2013skew}   Hadjesfandiari claims that an inhomogeneous  state $u$ of constant torsional deformation $\Chi(u)=0$ cannot exist. However, the conformal mapping, see Appendix \ref{aapendixconf}
  \begin{align}
\phi_c(x)=\frac{1}{2}\left(2\langle \axl\widehat{W},x\rangle \,x-\axl\widehat{W})\|x\|^2\right)+[\widehat{p}\, \id+\widehat{A}]. x+\widehat{b},
 \end{align}
where $\widehat{W},\widehat{A}\in \so(3)$, $\widehat{b}\in \mathbb{R}^3$, $\widehat{p}\in \mathbb{R}$ are arbitrary but constant is on the one hand inhomogeneous in the displacement $u=\phi_c-x$ but on the other hand gives precisely
 \begin{align}
 \Chi(u)=\Chi(\phi_c(x)-x)=\sym \nabla \curl \phi_c(x)=0.
 \end{align}
 For instance, a simple conformal displacement field is given by
 \begin{align}
 \notag
 \phi_c(x)&=\left(\begin{array}{c}
  2\, x_1^2 -(x_1^2+x_2^2+x_3^2)\\
   2\, x_1\, x_2 \\
    2\, x_1\, x_3
  \end{array}\right)=\left(\begin{array}{c}
   x_1^2 -(x_2^2+x_3^2)\\
   2\, x_1\, x_2 \\
    2\, x_1\, x_3
  \end{array}\right)\ \ \Rightarrow\ \ \nabla \phi_c(x)=\left(\begin{array}{ccc}
   2\, x_1& -2\, x_2& -2\, x_3\\
   2\, x_2&2\, x_1& 0 \\
    2\, x_3& 0&2\, x_1
  \end{array}\right)\\ &\Rightarrow\ \ \skew\nabla \phi_c(x)=\left(\begin{array}{ccc}
   0& -2\, x_2& -2\, x_3\\
   2\, x_2&0& 0 \\
    2\, x_3& 0&0
  \end{array}\right)\ \ \Rightarrow\ \ \axl\skew\nabla \phi_c(x)=\left(\begin{array}{c}
   0\\
   -2\, x_3 \\
    2\, x_2
  \end{array}\right) \\ &\Rightarrow\ \ \curl \phi_c(x)=\left(\begin{array}{c}
   0\\
   -\, x_3 \\
    \, x_2
  \end{array}\right)\ \ \Rightarrow\ \ \nabla \curl  \phi_c(x)=\left(\begin{array}{ccc}
   0& 0& 0\\
   0&0& -1 \\
    0& 1&0
  \end{array}\right)\Rightarrow\ \ \sym\nabla \curl\phi_c(x)=0.\notag
 \end{align}

\section{Existence and uniqueness of the solution in the Hadjesfandiari-Dargush formulation}\label{existHD}\setcounter{equation}{0}
In view of the previous discussion, we do not think that the Hadjesfandiari-Dargush formulation has a sound physical motivation. Nevertheless, mathematically it is possible to consider the parameter choice inherent in the Hadjesfandiari-Dargush formulation. Let us  consider for simplicity null boundary conditions. Hence, in the following we study the existence of the solution in the space
\begin{equation}
{\widetilde{\mathcal{X}}_0}\,{=}\,\big\{ u\,{\in}\,{H}^1_0(\Omega)\,|\, \curl u\in \, H(\curl; \Omega), \ \ (\id-n\otimes n).\,\curl u\big|_{\Gamma}=0\big\}.
\end{equation}
On ${\widetilde{\mathcal{X}}_0}$ we define the norm
\begin{equation}
\| u \|_{\widetilde{\mathcal{X}}_0}=\left(\|\nabla u\|^2_{L^2(\Omega)}+\| \skew \nabla\axl(\skew \nabla u)\|^2_{L^2(\Omega)} \right)^{\frac{1}{2}}=\left(\|\nabla u\|^2_{L^2(\Omega)}+\frac{1}{4}\| \curl \curl u\|^2_{L^2(\Omega)} \right)^{\frac{1}{2}},
\end{equation}
and the bilinear form
\begin{align}\label{proscalar}
 ((u,v))=\int_\Omega\bigg[& 2\mu\,\langle\sym \nabla u, \sym \nabla v \rangle+\lambda \tr(\nabla u)\,\tr( \nabla   v)\notag\\&+2\,\mu\,L_c^2\, \alpha_3\, \langle \skw[\nabla \axl(\skw \nabla u)],\skw[\nabla \axl(\skw \nabla v)]\rangle\bigg]\,dv \\
 =\int_\Omega\bigg[& 2\mu\,\langle\sym \nabla u, \sym \nabla v \rangle+\lambda \tr(\nabla u)\,\tr( \nabla   v)+2\,\mu\,L_c^2\, \alpha_3\, \langle \curl\curl u,\curl\curl v\rangle\bigg]\,dv,\notag
\end{align}
where
$u,v\in{\widetilde{\mathcal{X}}_0}$. Let us define the linear operator  $l:{\widetilde{\mathcal{X}}_0}\rightarrow\mathbb{R}$, describing the influence of external loads,
$
l(v)=\int_\Omega \langle f, v \rangle \, dv$  {for all} $\widetilde{w}\in{\mathcal{X}_0}.
$ We say that $w$ is a weak solution of the problem $(\mathcal{P})$ if and only if
\begin{equation}\label{wfdh}
((u,v))=l(v)  \ \  \text{ for all } \ \   v\in {\mathcal{X}_0}.
\end{equation}
A classical solution $u\in{\mathcal{X}_0}$
 of the problem $(\mathcal{P})$ is also a weak solution.

\begin{theorem}\label{thexdh}
Assume that
\begin{itemize}
\item[i)] the constitutive coefficients satisfy $\mu>0, \quad 3\, \lambda+2\mu>0, \quad \alpha_3\geq 0$;
\item[ii)] the loads satisfy the regularity condition $f\in L^2(\Omega)$.
\end{itemize}
Then there exists one and only one solution of the problem {\rm (\ref{wfdh})}.
\end{theorem}
\begin{proof}
In the case $\alpha_3=0$ we have the boundary value problem from classical elasticity. Further, we consider  the case $\boldsymbol{\alpha_3> 0}$.
The Cauchy-Schwarz inequality, the inequalities $(a\pm b)^2\leq 2(a^2+b^2)$ and the assumption upon the constitutive coefficients lead to
\begin{align}
((u,v))&\leq \dd \,C\, \|w\|_{{\widetilde{\mathcal{X}}_0}}\,\,\|\widetilde{w}\|_{{\widetilde{\mathcal{X}}_0}}\, ,
\end{align}
which means that $((\cdot,\cdot))$ is bounded. On the other hand, we have
\begin{align}
(({u},{u}))
=\int_\Omega\bigg[& 2\mu\,\|\sym \nabla u\|^2+\lambda\,[ \tr(\nabla u)]^2+2\,\mu\,L_c^2\, \alpha_3\, \|\skw[\nabla \axl(\skw \nabla u)]\|^2]\bigg]\, dv ,
\end{align}
for all
$u\in {\widetilde{\mathcal{X}}_0}$. Moreover, as a consequence of the properties  i) of the constitutive coefficients
 we have that there exist the positive constant $c$
\begin{align}
(({u},{u}))
&\geq \dd c\,\int _\Omega\biggl(\| \sym\nabla u\|^2+ \|\skw[\nabla \axl(\skw \nabla u)]\|^2\biggl)\, dv.
\end{align}
From  linearized elasticity we have   Korn's inequality \cite{Neff00b}, that is
\begin{align}\label{Korn}
\| \nabla  u\|_{L^2(\Omega)}\leq C \|\sym {\nabla}\, u\|_{L^2(\Omega)}\, ,
\end{align}
for all functions $u\in H_0^1(\Omega;\Gamma)$ with some constants $C>0$, for bounding the deformation of an elastic medium in terms of the symmetric strains.
Hence, using the Korn's inequality \eqref{Korn}, it results that there is a positive constant $C$ such that
\begin{align}
({u},{u})
&\geq \dd c\,\int _\Omega\biggl(\|\nabla u\|^2+ \|\skw[\nabla \axl(\skw \nabla u)] \|^2\biggl)\, dv=c\,\|u\|^2_{{\widetilde{\mathcal{X}}_0}}.
\end{align}
Hence our bilinear form $((\cdot,\cdot))$ is coercive.  The Cauchy-Schwarz inequality and the Poincar\'{e}-inequality  imply that the linear operator $l(\cdot)$ is bounded. By the Lax-Milgram theorem it follows that (\ref{wfdh}) has one and only one solution and the proof is complete.
\end{proof}

\begin{remark} The Lax-Milgram theorem used in the proof of the previous theorem also offers a continuous dependence result on the load $f$.   Moreover, the weak solution $u$  minimizes  on ${\mathcal{X}_0}$ the  energy functional
\begin{align}
I(u)=\int_\Omega\bigg[& 2\,\mu\,\|\sym \nabla u\|^2+\lambda\,[ \tr(\nabla u)]^2\notag+2\, \mu\,L_c^2\, \alpha_3\, \|\skw(\Curl (\sym \nabla u))\|^2-\langle f,u\rangle \bigg]\, dv.\notag
\end{align}
\end{remark}

\section{The constrained Cosserat formulation of the  Hadjesfandiari-Dargush model: well posedness of a degenerate Cosserat model}\setcounter{equation}{0}

Similarly  to the classical indeterminate couple stress model, also the Hadjesfandiari-Dargush formulation can be obtained as a constrained Cosserat model. We only need to adapt the curvature energy. We consider the replacement ${\skew \,\nabla u \mapsto \overline{A}}\in \so(3)$, and we obtain the energy
\begin{align}
\mathcal{W}&=2\,\mu\,\|{\rm sym} \nabla u\|^2+\lambda\,[ \tr(\nabla u)]^2+\mu_c\,\|\skw \nabla u-\overline{A}\|^2+2\,\mu\,L_c^2\,\|\skew\nabla \axl(\overline{A})\|^2\\
&=2\,\mu\,\|{\rm sym} \nabla u\|^2+\lambda\,[ \tr(\nabla u)]^2+\mu_c\,\|\skw \nabla u-\overline{A}\|^2+2\,\mu\,L_c^2\,\|\curl \axl(\overline{A})\|^2\notag
\end{align}
for the Cosserat model.

Using the usual procedure, it follows that there exists a unique solution $(u,\overline{A})$ of the corresponding minimization problem, i.e. to find the minimum of the energy
\begin{align}
I(u)=\int_\Omega\bigg[& 2\,\mu\,\|{\rm sym} \nabla u\|^2+\lambda\,[ \tr(\nabla u)]^2+\mu_c\|\skw \nabla u-\overline{A}\|^2+2\,\mu\,L_c^2\,\|\curl \axl(\overline{A})\|^2\\
&-\langle f,u\rangle-\langle \axl(\mathfrak{M}),\axl(\overline{A})\rangle \bigg]\, dv,\notag
\end{align}
where $f:\Omega\rightarrow\mathbb{R}^3$ and $\mathfrak{M}:\Omega\rightarrow\mathbb{R}^{3\times 3}$ are prescribed, such that
$u\in H_0^1(\Omega)$ and $\axl(\overline{A})\in H(\Curl;\Omega)$, $\axl(\overline{A})\times n\big|_{\Gamma}=0$.

\begin{figure}
\setlength{\unitlength}{1mm}
\begin{center}
\begin{picture}(10,30)
\thicklines

\put(-40,18){\oval(86,43)}
\put(-81,36){\footnotesize{\bf \ \ \ $\boldsymbol{\skew \nabla u \mapsto \overline{A}}\in \so(3)$, degenerate Cosserat model }}
\put(-81,32){\footnotesize{ ${\sigma\sim 2\,\mu\, \sym\nabla u+2\, \mu_c \skw(\nabla u-\overline{A})\not\in {\rm Sym}(3)}$}}
\put(-81,28){\footnotesize{ $m=2\,\mu\,L_c^2\, \skew\nabla \axl(\overline{A})\in \so(3)$}}
\put(-81,24){\footnotesize{ $\mathcal{E}\sim\mu\,\|{\rm sym} \nabla u\|^2+\mu_c\|\skw \nabla u-\overline{A}\|^2+\mu\,L_c^2\,\|\skew\nabla \axl(\overline{A})\|^2$}}
\put(-81,20){\footnotesize{ \bf invariant under: \!$u\mapsto u+\overline{W}.x+\overline{b}, \, \overline{b}\in \mathbb{R}^3$}}
\put(-81,16){\footnotesize{\bf \qquad\qquad\quad \qquad\    $A\mapsto \overline{A}+{\overline{W}}, \,\,\,\overline{W}\in \so(3)$}}
\put(-81,12){\footnotesize{ {\bf well-posed: } $u\in H^1(\Omega)$, $\axl(\overline{A})\in H({\rm curl};\Omega)$, }}
\put(-81,8){\footnotesize{ {\bf degenerate curvature energy}, }}
\put(-81,4){\footnotesize{ \bf   $\boldsymbol{3+2=5}$ geometric bc: $u\big|_\Gamma, \ \  \langle \axl(\overline{A}),\tau_\alpha\rangle\big|_\Gamma$}}
\put(-81,0){\footnotesize{ \bf   $\boldsymbol{3+2=5}$ traction bc: $\sigma.n\big|_{\partial \Omega\setminus \Gamma},\ \ \ \langle m.n,\tau_\alpha\rangle\big|_{\partial \Omega\setminus \Gamma}$ }}

\put(3,18){\vector(1,0){19}}
\put(8,28){\footnotesize{$\nabla u=\overline{A}$}}
\put(8,24){\footnotesize{constrain}}
\put(8,20){\footnotesize{$\mu_c\rightarrow\infty$}}

\put(57,19){\oval(69,44)}
\put(25,36){\footnotesize{ \bf Hadjesfandiari-Dargush model}}
\put(25,32){\footnotesize{ \bf $\sigma\sim 2\,\mu\, \sym\nabla u\in {\rm Sym}(3)$}}
\put(25,28){\footnotesize{ $\widetilde{ {m}}=2\,\mu\,L_c^2\,  \skew \nabla \axl(\skew \nabla  u)\in\so(3)$}}
\put(25,24){\footnotesize{ $\widetilde{\tau}=2\,\mu\,L_c^2\,  \anti[\Div(\skew\nabla \axl(\skew \nabla  u))]\in\so(3)$}}
\put(25,20){\footnotesize{ $\mathcal{E}\sim\mu\,\|{\rm sym} \nabla u\|^2+\mu\,L_c^2\,\|\skew \nabla \axl(\skew \nabla  u)\|^2$}}
\put(25,16){\footnotesize{ \bf invariant under:  $u\mapsto u+\overline{W}.x+\overline{b},  \overline{b}\in \mathbb{R}^3$}}
\put(25,12){\footnotesize{ \bf \ \quad\quad\quad\quad \qquad\    $\nabla u\mapsto \nabla u+\overline{W}, {\overline{W}}\in \so(3)$}}
\put(25,8){\footnotesize{{\bf well-posed:}  $u\in H^1(\Omega)$,  $\curl u\in H(\curl;\Omega)$  }}
\put(25,4){\footnotesize{\bf 5 geometric bc}:  $u\big|_\Gamma$, $\langle \curl u,\tau_\alpha\rangle\big|_{\Gamma}$}
\put(25,0){\footnotesize{\bf  5 traction bc}}

\end{picture}
\end{center}
\caption{A possibility of lifting the 4th.-order indeterminate  couple stress model to a 2nd.-order micromorphic or Cosserat-type formulation formulation. Here $\tau_\alpha$, $\alpha=1,2$ denote two independent tangential vectors on the boundary. }\label{limitmodel}
\end{figure}
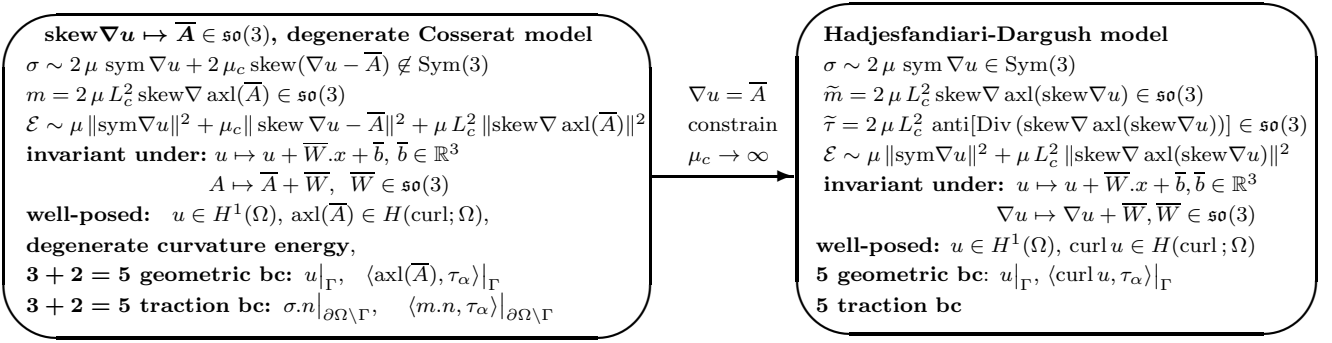

\section{Conclusion}

First, Hadjesfandiari and Dargush reject the notion of additional independent degrees of freedom, which
is purely arbitrary \cite[p. 2496]{hadjesfandiari2011couple}.
 The ``consistent" theory proposed by Hadjesfandiari and Dargush ``resolving all difficulties...", taking only the skew-symmetric
part of the curvature tensor $\nabla \curl \, u$,
 is simply a special case corresponding to some vanishing moduli in the theory. As we have seen, this  choice  is a restriction, but not a necessity.

 In summary, Hadjesfandiari and Dargush have raised \cite[p.~17]{hadjesfandiari2011couple} three  major concerns in the indeterminate couple stress model (see also \cite{hadjesfandiari2013skew}):
 \begin{quote}
 \begin{itemize}
 \item[1)] \textit{The body-couple is present in the constitutive relations for the [total] force-stress tensor in the {\rm [indeterminate couple stress]} theory}.
\item[2)] \textit{The spherical part of the couple-stress tensor is indeterminate, because the curvature
tensor {\rm [$\widetilde{k}=\nabla {\rm curl} \,u$]} is deviatoric}.
\item[3)] \textit{The boundary conditions are inconsistent, because the normal component of moment
traction {\rm [$\langle n,\widetilde{m}.n\rangle$]} appears in the formulation {\rm [of the force stress tensor]}}.
 \end{itemize}
\end{quote}

   Regarding the above three ``serious
inconsistencies" presented in \cite[p. 17]{hadjesfandiari2011couple} we may answer:
\begin{quote}
  \begin{itemize}
 \item[1)] This is of course not inconsistent.
It is well-known that the Cauchy-like total force stress tensor $\sigma-\widetilde{\tau}$  is not the constitutive stress. The
constitutive stresses and couples are those arising from the virtual work principle
\eqref{gradeq211} by fixing $\sym \nabla u$ and $\nabla \curl u$ as independent constitutive quantities\footnote{However, this does not imply that we may independently prescribe
 $u$ and $\curl u$ on the boundary.}. More precisely,   the
constitutively dependent quantities are the energetic conjugates of $\frac{\rm d}{\rm dt}\sym \nabla u$ and $\frac{\rm d}{\rm dt}\nabla \curl u$, respectively. We have to note that the constitutive dependent quantities are not the energetic conjugates of $\frac{\rm d}{\rm dt}u$ and $\frac{\rm d}{\rm dt}\curl u$, respectively.

 The relation between  the Cauchy-like stress and constitutive stress
indeed involve the volume simple double and triple forces (see \cite{Germain1} in
French, first part on the classical theory).
Note that in the Hadjesfandiari and Dargush-formulation \eqref{HBBC}, the boundary conditions would acquire the same form and meaning as in the classical format: the total force stress tensor $\sigma-\widetilde{\tau}$ would be the Cauchy stress tensor and the curvature of the surface normal would not intervene in the traction boundary conditions. However, the incorrect boundary conditions  used by them  \eqref{HBBC} are not any more in the completely independent decomposition form  \eqref{HBBC1} of the boundary conditions;
\item[2)] The indeterminacy of the spherical part of the couple stress is not inconsistent.
Like the pressure in an incompressible body, it is indeterminate in the local
constitutive law but can be found from the boundary conditions after solving the equilibrium equations.
 It is a reaction stress, as it is well-known in the theory of continua with
internal constraints.
 Should we say that the theory of incompressible bodies is ``inconsistent"?  Surely
not, it is mathematically clear even though it may induce
 computational difficulties.
 \item[3)] As we have shown in this paper, see also \cite{MadeoGhibaNeffMunchKM}, the boundary conditions used by Hadjesfandiari and Dargush \cite{hadjesfandiari2011couple} are incomplete, therefore no inconsistency occurs,  see again \eqref{HBBC1}  vs. \eqref{HBBC}.
 \end{itemize}
\end{quote}

Finally, we like to mention that  the extra constitutive parameter $\eta^\prime$ of Grioli's model does not intervene in the field
 partial differential equations.
It will arise through the boundary conditions.
 We do not see which mechanical principle is violated by that. In summary, there is  a fully consistent version of the indeterminate-couple stress model with only 3 constitutive parameters. This is the modified couple stress model, see \cite{Neff_Jeong_IJSS09,NeffGhibaMadeoMunch} and a novel variant of it is recently discussed in \cite{NeffGhibaMadeoMunch}.

\section*{Acknowledgement} We are grateful to Ali Reza Hadjesfandiari (University at Buffalo) and Gary F. Dargush (University at Buffalo) for sending us the paper \cite{hadjesfandiari2014evo} prior to publication. We would like to thank Samuel Forest (CNRS Mines ParisTech) for detailed discussions. Ionel-Dumitrel Ghiba acknowledges support from the Romanian National Authority for Scientific Research (CNCS-UEFISCDI), Project No. PN-II-ID-PCE-2011-3-0521.

\bibliographystyle{plain} 

\begin{thebibliography}{10}

\bibitem{Adams75}
R.A. Adams.
\newblock {\em Sobolev {S}paces.}, volume~65 of {\em Pure and {A}pplied
  {M}athematics}.
\newblock Academic Press, London, 1. edition, 1975.

\bibitem{Aero61}
E.L. Aero and E.V. Kuvshinskii.
\newblock Fundamental equations of the theory of elastic media with
  rotationally interacting particles.
\newblock {\em Soviet Physics-Solid State}, 2:1272--1281, 1961.

\bibitem{aifantis2011gradient}
E.C. Aifantis.
\newblock On the gradient approach--relation to {E}ringen's nonlocal theory.
\newblock {\em Int. J. Eng. Sci.}, 49(12):1367--1377, 2011.

\bibitem{dell2012beyond}
F.~dell'Isola, G.~Sciarra, and A.~Madeo.
\newblock {\em Beyond {Euler-Cauchy Continua: The} structure of contact actions
  in {N-th} gradient generalized continua: a generalization of the {Cauchy}
  tetrahedron argument}.
\newblock CISM Lecture Notes C-1006, Chap.2. Springer, 2012.

\bibitem{NthGrad}
F.~dell'Isola, P.~Seppecher, and A.~Madeo.
\newblock How contact interactions may depend on the shape of {C}auchy cuts in
  {N}th gradient continua: approach ``{\'a} la {d'A}lembert".
\newblock {\em Z. Angew. Math. Phys.}, 63(6):1119--1141, 2012.

\bibitem{Germain1}
P.~Germain.
\newblock The method of virtual power in continuum mechanics. {Part 2}:
  {Microstructure}.
\newblock {\em SIAM J. Appl. Math.}, 25:556--575, 1973.

\bibitem{NeffGhibaMadeoMunch}
I.D. Ghiba, P.~Neff, A.~Madeo, and I.~M\"unch.
\newblock A variant of the linear isotropic indeterminate couple stress model
  with symmetric local force-stress, symmetric nonlocal force-stress, symmetric
  couple-stresses and complete traction boundary conditions.
\newblock {\em submitted}, Preprint arXiv:1504.00868, 2015.

\bibitem{Raviart79}
V.~Girault and P.A. Raviart.
\newblock {\em Finite Element Approximation of the {N}avier-{S}tokes
  Equations.}, volume 749 of {\em Lect. Notes Math.}
\newblock Springer, Heidelberg, 1979.

\bibitem{Grioli60}
G.~Grioli.
\newblock Elasticit\'a asimmetrica.
\newblock {\em Ann. Mat. Pura Appl., Ser. IV}, 50:389--417, 1960.

\bibitem{grioli2003microstructures}
G.~Grioli.
\newblock Microstructures as a refinement of {Cauchy theory. Problems} of
  physical concreteness.
\newblock {\em Cont. Mech. Thermodyn.}, 15(5):441--450, 2003.

\bibitem{gurtin2010mechanics}
M.E. Gurtin, E.~Fried, and L.~Anand.
\newblock {\em The mechanics and thermodynamics of continua}.
\newblock Cambridge University Press, 2010.

\bibitem{hadjesfandiari2010polar}
A.~Hadjesfandiari and G.F. Dargush.
\newblock Polar continuum mechanics.
\newblock {\em Preprint arXiv:1009.3252}, 2010.

\bibitem{hadjesfandiari2011couple}
A.~Hadjesfandiari and G.F. Dargush.
\newblock Couple stress theory for solids.
\newblock {\em Int. J. Solids Struct.}, 48(18):2496--2510, 2011.

\bibitem{hadjesfandiari2013fundamental}
A.~Hadjesfandiari and G.F. Dargush.
\newblock Fundamental solutions for isotropic size-dependent couple stress
  elasticity.
\newblock {\em Int. J. Solids Struct.}, 50(9):1253--1265, 2013.

\bibitem{hadjesfandiari2013skew}
A.R. Hadjesfandiari.
\newblock On the skew-symmetric character of the couple-stress tensor.
\newblock {\em Preprint arXiv:1303.3569}, 2013.

\bibitem{Dargush}
A.R. Hadjesfandiari and G.F. Dargush.
\newblock Couple stress theory for solids.
\newblock {\em Int. J. Solids Struct.}, 48:2496--2510, 2011.

\bibitem{hadjesfandiari2014evo}
A.R. Hadjesfandiari and G.F. Dargush.
\newblock Evolution of generalized couple-stress continuum theories: a critical
  analysis.
\newblock {\em Preprint arXiv:1501.03112}, 2015.

\bibitem{Neff_JeongMMS08}
J.~Jeong and P.~Neff.
\newblock Existence, uniqueness and stability in linear {C}osserat elasticity
  for weakest curvature conditions.
\newblock {\em Math. Mech. Solids}, 15(1):78--95, 2010.

\bibitem{Koiter64}
W.T. Koiter.
\newblock Couple stresses in the theory of elasticity {I,II}.
\newblock {\em Proc. Kon. Ned. Akad. Wetenschap}, B 67:17--44, 1964.

\bibitem{lazar2006note}
M.~Lazar and G.A. Maugin.
\newblock A note on line forces in gradient elasticity.
\newblock {\em Mech. Research Comm.}, 33(5):674--680, 2006.

\bibitem{Leis86}
R.~Leis.
\newblock {\em Initial Boundary Value problems in Mathematical Physics}.
\newblock Teubner, Stuttgart, 1986.

\bibitem{lubarda2003effects}
V.A. Lubarda.
\newblock The effects of couple stresses on dislocation strain energy.
\newblock {\em Int. J. Solids Struct.}, 40(15):3807--3826, 2003.

\bibitem{MadeoGhibaNeffMunchKM}
A.~Madeo, I.D. Ghiba, P.~Neff, and I.~M\"unch.
\newblock Incomplete traction boundary conditions in
  {Grioli-Koiter-Mindlin-Toupin}'s indeterminate couple stress model.
\newblock {\em in preparation}, 2015.

\bibitem{maugin1980method}
G.A. Maugin.
\newblock The method of virtual power in continuum mechanics: application to
  coupled fields.
\newblock {\em Acta Mech.}, 35(1-2):1--70, 1980.

\bibitem{MauginVirtualPowers}
G.A. Maugin.
\newblock The principle of virtual power: from eliminating metaphysical forces
  to providing an efficient modelling tool. {In} memory of {Paul Germain}
  (1920-2009).
\newblock {\em Cont. Mech. Thermodyn.}, 25:127--146, 2013.

\bibitem{Mindlin65}
R.D. Mindlin.
\newblock Second gradient of strain and surface tension in linear elasticity.
\newblock {\em Int. J. Solids Struct.}, 1:417--438, 1965.

\bibitem{Mindlin68}
R.D. Mindlin and N.N. Eshel.
\newblock On first strain-gradient theories in linear elasticity.
\newblock {\em Int. J. Solids Struct.}, 4:109--124, 1968.

\bibitem{Mindlin62}
R.D. Mindlin and H.F. Tiersten.
\newblock Effects of couple stresses in linear elasticity.
\newblock {\em Arch. Rat. Mech. Anal.}, 11:415--447, 1962.

\bibitem{Neff00b}
P.~Neff.
\newblock On {K}orn's first inequality with nonconstant coefficients.
\newblock {\em Proc. Roy. Soc. Edinb. A}, 132:221--243, 2002.

\bibitem{Neff_Cosserat_plasticity05}
P.~Neff.
\newblock A finite-strain elastic-plastic {C}osserat theory for polycrystals
  with grain rotations.
\newblock {\em Int. J. Eng. Sci.}, 44:574--594, 2006.

\bibitem{Neff_Forest_jel05}
P.~Neff and S.~Forest.
\newblock A geometrically exact micromorphic model for elastic metallic foams
  accounting for affine microstructure. {Modelling}, existence of minimizers,
  identification of moduli and computational results.
\newblock {\em J. Elasticity}, 87:239--276, 2007.

\bibitem{MadeoGhibaNeffMunchCRM}
P.~Neff, I.D. Ghiba, A.~Madeo, and I.~M\"unch.
\newblock Correct traction boundary conditions in the indeterminate couple
  stress model.
\newblock {\em submitted}, Preprint arXiv:1504.00448, 2015.

\bibitem{Neff_Jeong_Conformal_ZAMM08}
P.~Neff and J.~Jeong.
\newblock A new paradigm: the linear isotropic {C}osserat model with
  conformally invariant curvature energy.
\newblock {\em Z. Angew. Math. Mech.}, 89(2):107--122, 2009.

\bibitem{Neff_Jeong_bounded_stiffness09}
P.~Neff, J.~Jeong, and A.~Fischle.
\newblock Stable identification of linear isotropic {C}osserat parameters:
  bounded stiffness in bending and torsion implies conformal invariance of
  curvature.
\newblock {\em Acta Mech.}, 211(3-4):237--249, 2010.

\bibitem{Neff_Paris_Maugin09}
P.~Neff, J.~Jeong, I.~M\"unch, and H.~Ramezani.
\newblock Linear {C}osserat {E}lasticity, {C}onformal {C}urvature and {B}ounded
  {S}tiffness.
\newblock In G.A. Maugin and V.A. Metrikine, editors, {\em Mechanics of
  Generalized Continua. One hundred years after the Cosserats}, volume~21 of
  {\em Advances in Mechanics and Mathematics}, pages 55--63. Springer, Berlin,
  2010.

\bibitem{Neff_Jeong_IJSS09}
P.~Neff, J.~Jeong, and H.~Ramezani.
\newblock Subgrid interaction and micro-randomness - novel invariance
  requirements in infinitesimal gradient elasticity.
\newblock {\em Int. J. Solids Struct.}, 46(25-26):4261--4276, 2009.

\bibitem{Park07}
S.K. Park and X.L. Gao.
\newblock Variational formulation of a simplified strain gradient elasticity
  theory and its application to a pressurized thick-walled cylinder problem.
\newblock {\em Int. J. Solids Struct.}, 44:7486--7499, 2007.

\bibitem{park2008variational}
S.K. Park and X.L. Gao.
\newblock Variational formulation of a modified couple stress theory and its
  application to a simple shear problem.
\newblock {\em Z. Angew. Math. Mech.}, 59:904--917, 2008.

\bibitem{TheseSeppecher}
P.~Seppecher.
\newblock {\em Etude d\textquoteright une Modelisation des Zones Capillaires
  Fluides: Interfaces et Lignes de Contact}.
\newblock {Ph.D-Thesis}, Ecole Nationale Superieure de Techniques Avancees,
  Universit\'e Pierre et Marie Curie, Paris, 1987.

\bibitem{Sokolowski72}
M.~Sokolowski.
\newblock {\em Theory of {C}ouple {S}tresses in {B}odies with {C}onstrained
  {R}otations.}, volume~26 of {\em International Center for Mechanical Sciences
  CISM: Courses and Lectures}.
\newblock Springer, Wien, 1972.

\bibitem{Toupin62}
R.A. Toupin.
\newblock Elastic materials with couple stresses.
\newblock {\em Arch. Rat. Mech. Anal.}, 11:385--413, 1962.

\bibitem{Toupin64}
R.A. Toupin.
\newblock Theory of elasticity with couple stresses.
\newblock {\em Arch. Rat. Mech. Anal.}, 17:85--112, 1964.

\bibitem{Yang02}
F.~Yang, A.C.M. Chong, D.C.C. Lam, and P.~Tong.
\newblock Couple stress based strain gradient theory for elasticity.
\newblock {\em Int. J. Solids Struct.}, 39:2731--2743, 2002.

\end{thebibliography}
\addcontentsline{toc}{section}{References}

\begin{footnotesize}

\appendix

\section{Conformal invariance of the curvature energy and group theoretic arguments in favor of the modified couple stress theory}\setcounter{equation}{0}\label{aapendixconf}

This section is taken from \cite{NeffGhibaMadeoMunch} and included here for this contribution to be rather self-contained. An infinitesimal conformal mapping \cite{Neff_Jeong_Conformal_ZAMM08,Neff_Jeong_IJSS09} preserves (to first order) angles and shapes of infinitesimal figures. The included inhomogeneity is therefore only a global feature of the mapping. There is locally no shear-type deformation. Therefore it seems natural to require that the second gradient model should not ascribe energy to such deformation modes.

A map $\phi_c:\mathbb{R}^3\rightarrow\mathbb{R}^3$ is infinitesimal conformal if and only if its Jacobian satisfies pointwise $\nabla \phi_c(x)\in \mathbb{R}\cdot \id+\so(3)$, where $\mathbb{R}\cdot \id+\so(3)$ is the conformal Lie-algebra. This implies \cite{Neff_Jeong_Conformal_ZAMM08,Neff_Jeong_IJSS09,Neff_Jeong_bounded_stiffness09} the representation
\begin{align}
\phi_c(x)=\frac{1}{2}\left(2\langle \axl\overline{W},x\rangle \,x-\axl\overline{W}\|x\|^2\right)+[\widehat{p}\,\id+\widehat{A}]. x+\widehat{b}\,,
\end{align}
where $\overline{W},\widehat{A}\in \so(3)$, $\widehat{b}\in \mathbb{R}^3$, $\widehat{p}\in \mathbb{R}$ are arbitrary given constants.
For the infinitesimal conformal mapping $\phi_c$ we note
\begin{align}\label{conf}
\begin{array}{ll}
\nabla \phi_c(x)=[\langle \axl\overline{W},x\rangle+\widehat{p}]\, \id +\anti(\overline{W}.\, x)+\widehat{A},&\qquad
{\rm div}\,\phi_c(x)=\tr[\nabla \phi_c(x)]=3\, [\langle \axl\overline{W},x\rangle+\widehat{p}],\vspace{1.2mm}\\
\skew  \nabla \phi_c(x)=\anti(\overline{W}.\, x)+\widehat{A},&\qquad
\sym \nabla \phi_c(x)=[\langle \axl\overline{W},x\rangle+\widehat{p}]\, \id,\vspace{1.2mm}\\
\dev\sym \nabla \phi_c(x)=0,&\qquad
\nabla \curl \phi_c(x)=2\, \overline{W}\in \so(3),\vspace{1.2mm}\\
\sym \nabla\curl \phi_c(x)=0,&\qquad
\skew\,\nabla \curl \phi_c(x)=2\, \overline{W}.
\end{array}
\end{align}
These relations  are easily established. By {\bf conformal invariance} of the curvature energy term we  mean that the curvature energy vanishes on infinitesimal conformal mappings. This is equivalent to
\begin{align}
W_{\rm curv}(D^2 \phi_c)=0\qquad \text{for all conformal maps} \quad \phi_c,
\end{align}
or in terms of the second order couple stress tensor $\widetilde{m}:=D_{\nabla \curl u}W_{\rm curv}(\nabla \curl u)$,
\begin{align}
\widetilde{m}(D^2 \phi_c)=0 \qquad \text{for all conformal maps} \quad \phi_c.
\end{align}
The classical linear elastic energy still ascribes energy to such a deformation mode, but only related to the bulk modulus, i.e.,
\begin{align}
W_{\rm lin}(\nabla \phi_c)=\underbrace{\mu\, \|\dev\sym \nabla \phi_c\|^2}_{=0}+\frac{2\,\mu+3\,\lambda}{2} \, [\tr(\nabla  \phi_c)]^2=\frac{2\,\mu+3\,\lambda}{2} \, [\tr(\nabla  \phi_c)]^2.
\end{align}
In case of a classical infinitesimal perfect plasticity formulation with von Mises deviatoric flow rule, conformal mappings are precisely those inhomogeneous mappings, that do not lead to plastic flow \cite{Neff_Cosserat_plasticity05}, since the deviatoric stresses remain zero: $\dev \sym \nabla \phi_c=0$.

In that perspective
\[\framebox[1.1\width]{\textit{conformal mappings are ideally elastic transformations and should not lead to moment stresses.}}
\]

The underlying additional invariance property of the modified couple stress theory is precisely conformal invariance. In the modified couple stress model, these deformations are free of size-effects, while e.g. the Hadjesfandiari and Dargush choice would describe size-effects. In other words, the generated couple stress tensor $\widetilde{m}$ in the modified couple stress model is zero for this inhomogeneous deformation mode, while in the  Hadjesfandiari and Dargush choice $\widetilde{m}$ is constant and skew-symmetric\footnote{This observation is a further development in understanding why the Hadjesfandiari and Dargush \cite{hadjesfandiari2010polar,hadjesfandiari2011couple,hadjesfandiari2014evo} choice is rather meaningless, while mathematically not forbidden.}.
\end{footnotesize}
\end{document}